\documentclass[reqno, 11pt]{amsart}

\usepackage{amsmath,amssymb,amsthm,amsfonts}
\usepackage{mathtools} 
\usepackage[normalem]{ulem}
\usepackage{mathrsfs}
\usepackage{hyperref}

\newtheorem{theorem}{Theorem}[section]
\newtheorem{lemma}[theorem]{Lemma}
\newtheorem{prop}[theorem]{Proposition}
\newtheorem{cor}[theorem]{Corollary}

\theoremstyle{remark}
\newtheorem*{remark}{Remark}

\newcommand{\e}{\mathrm{e}} 
\newcommand{\N}{\mathbb{N}}
\newcommand{\R}{\mathbb{R}}
\newcommand{\C}{\mathbb{C}}
\newcommand{\X}{\mathbb X}

\newcommand{\be}{\begin{equation}}
\newcommand{\ee}{\end{equation}}

\def\1{{\mathchoice {1\mskip-4mu\mathrm l}      
{1\mskip-4mu\mathrm l}
{1\mskip-4.5mu\mathrm l} {1\mskip-5mu\mathrm l}}}

\begin{document}

\title{A Multi-Body Dobrushin--Sokal Criterion -- Part I} 
\date{August 16, 2025}

\author[J.~P.~Neumann]{Jan Philipp Neumann}
\address{Mathematisches Institut, Ludwig-Maximilians-Universit{\"a}t, 80333 M{\"u}nchen, Germany.}
\email{neumann@math.lmu.de}

\begin{abstract}
	We derive a sufficient condition for zero-freeness of partition functions applicable to lattice gases with possibly complex-valued multi-body interactions. This includes the case of hard-core interactions and, in particular, generalises recent results by Galvin et al.\ (2024) and Bencs--Buys (2025) on zero-free polydiscs of hypergraph independence polynomials. We provide two proofs: the first generalises the inductive approach of Bencs and Buys; the second employs the Kirkwood--Salsburg hierarchy. Notably, the central argument of the second proof uses of a certain partition scheme for coverings and, as a by-product, we obtain a direct improvement of Gallavotti and Miracle-Sol{\'e}'s (1968) bounds for the Kirkwood--Salsburg operator.
	
	\bigskip 
		
    \noindent \emph{Mathematics Subject Classification}: 05C31, 05C65, 05C69, 82B20.
    
	\medskip     
    
    \noindent \emph{Keywords}: Classical lattice gas; multi-body interactions; grand partition function; effective interactions; hypergraph independence polynomial; Kirkwood--Salsburg hierarchy.
\end{abstract}

\maketitle

\section{Introduction} \label{sec:introduction}

\noindent Recently, Bencs and Buys proved in \cite{ref:bencs-buys} that, given a finite hypergraph of maximal degree $\Delta \geq 2$, its independence polynomial has no zeroes in the activity polydisc with uniform radius $\frac{(\Delta-1)^{\Delta-1}}{\Delta^\Delta}$. Absent any information beyond $\Delta$, this radius is optimal and improves upon the lower bound $\frac{\Delta^\Delta}{(\Delta+1)^{\Delta+1}}$ given by Galvin et.\ al.\ in \cite{ref:galvin-et-al} so shortly before that the latter already acknowledges the former authors' improvement, then still in preprint. For standard graphs, these radii are well-known, the optimal one going back to Shearer's 1985 paper motivated by the Lov{\'a}sz local lemma, see, e.g., \cite{ref:scott-sokal}.

Beyond motivations concerning combinatorics and algorithmic efficiency, all of the articles \cite{ref:bencs-buys, ref:galvin-et-al, ref:scott-sokal} mention, to a more or less significant extent, the relevance from the view of equilibrium statistical mechanics. The connection lies in the interpretation of the independence polynomial as the partition function of a classical lattice gas with hard-core self-repulsion, whose zero-free activity domains are related to parameter regimes without phase transitions according to Lee and Yang, see \cite{ref:lee-yang-i, ref:lee-yang-ii}. While the independence polynomial only covers lattice gases with pure (multi-body) hard-core repulsions, the Lee--Yang relation extends to situations with arbitrary interactions and the present article carries the argument of Bencs and Buys through this extension.

The result, Theorem~\ref{thm:dobrushin-Z-non-zero-finite}, gives a new sufficient criterion for zero-freeness of general lattice gas partition functions at sufficiently small activity. Not only is this criterion, to the best of our knowledge, novel in its range of admissible lattice gas interactions but, by providing two proofs, we can unify the recent progress on the combinatorics side with improvements to long-established approaches on the statistical physics side, extending the comparison in \cite{ref:scott-sokal} to possibly attractive or even complex multi-body interactions. In fact, for pairwise repulsions, Theorem~\ref{thm:dobrushin-Z-non-zero-finite} becomes a criterion found in the former reference, where credit is given to Dobrushin (\cite{ref:dobrushin-semi-invariants, ref:dobrushin-saint-flour}) and Sokal (\cite{ref:sokal}).

The first proof spans Section~\ref{sec:proof} and is the aforementioned generalisation of some of Bencs and Buys' combinatorics in \cite{ref:bencs-buys} to complex-valued interactions. The subsequent Section~\ref{sec:independence-polynomial} then connects our streamlined notation back to the setting of hard-core interactions and independence polynomials. The result of Galvin et.\ al., \cite[Theorem~1]{ref:galvin-et-al}, see Corollary~\ref{cor:galvin-et-al} here, is an immediate consequence of our main result. Regarding Bencs and Buys' result, \cite[Theorem~1.1]{ref:bencs-buys}, here Theorem~\ref{thm:bencs-buys}, we slightly modify the proof since a few aspects of their original proof are omitted in our generalisation. This modified proof closely follows the treatment in \cite{ref:scott-sokal} for standard graphs.

In Section~\ref{sec:alternative}, we provide a second proof of our main result based on the classical Kirkwood--Salsburg hierarchy, cf., e.g., Ruelle's book \cite{ref:ruelle}, with a non-standard twist. The more intricate parts of Bencs and Buys' construction, or our generalisation thereof, become obsolete in this alternative approach but we refrain from proving \cite[Theorem~1.1]{ref:bencs-buys} a second time here. We essentially use the ansatz described for Gibbs point processes in \cite{ref:jansen-cluster, ref:jansen-kolesnikov} by Jansen and Kolesnikov. This avoids the need for the strict contractivity of the Kirkwood--Salsburg operator in the classical Banach fixed point argument as employed in \cite{ref:gallavotti-miracle-sole, ref:gallavotti-miracle-sole-robinson, ref:ruelle}. The heart of our argument is a partition scheme that sorts coverings of finite sets according to in a certain sense minimal sub-coverings. Reading this argument as a bound on the Kirkwood--Salsburg operator directly improves the bound given by Gallavotti and Miracle-Sol{\'e} in \cite{ref:gallavotti-miracle-sole} which forms the basis of the analyticity results there and in \cite{ref:gallavotti-miracle-sole-robinson, ref:ruelle}.

Although we do not appeal to it here, the ansatz taken from \cite{ref:jansen-cluster, ref:jansen-kolesnikov} implies the absolute convergence of the Mayer cluster expansions for log-partition functions and (infinite-volume) correlations. To the interested reader, let us, for now, at least mention the following: the case of pairwise interactions has been investigated quite extensively, see, e.g., \cite{ref:scott-sokal, ref:ueltschi, ref:poghosyan-ueltschi, ref:jansen-cluster, ref:jansen-kolesnikov} and the references therein; the case of multi-body interactions is combinatorially harder and therefore comparatively underdeveloped but one may have a look at the references listed in the introduction of \cite{ref:galvin-et-al} or polymer expansion approaches such as \cite{ref:procacci-scoppola, ref:nguyen-fernandez}.

Regarding the organisation of the article, Section~\ref{sec:results} introduces the basic model-defining quantities, states our main result, Theorem~\ref{thm:dobrushin-Z-non-zero-finite}, and closes with a short discussion of its relation to some well-established comparable results from the literature. This is followed by a brief but essential extension of our notation in Section~\ref{sec:preliminaries}. The rest is as outlined above: the first proof of our main result spans Section~\ref{sec:proof}; Section~\ref{sec:independence-polynomial} then returns to the special case of independence polynomials to illustrate the results from \cite{ref:galvin-et-al, ref:bencs-buys}; finally, the alternative proof, based on the Kirkwood--Salsburg hierarchy, is featured in Section~\ref{sec:alternative}.

\section{Results} \label{sec:results}

\noindent Throughout the paper, we fix a \emph{lattice} $\X$, a finite or countably infinite set whose elements are called \emph{sites}. In the context of this article, a \emph{configuration} assigns to each site $x \in \X$ one of two states, \emph{unoccupied} or \emph{occupied}, and is formally regarded as an element of
\[
	\mathbf F : = \{X \Subset \X\} : = \{X \subset \X \mid |X| < \infty\},
\]
the set of finite subsets of $\X$. Infinite configurations are irrelevant within the scope of this article.

The propensity of individual sites to be occupied is expressed by a function $z : \X \to \C$, called \emph{activity} (or \emph{fugacity}), and we assign to each configuration $X \Subset \X$ the multiplicative a priori weight
\[
	z^X : = \prod_{x \in X} z(x),
\]
where we treat products over empty index sets as equal to $1$. Of course, the words ``propensity'' and ``weight'' are somewhat inadequate unless $z \geq 0$ or, say, $z = \exp(\beta \mu)$ for some \emph{chemical potential} $\mu : \X \to \R$ and an \emph{inverse temperature} $\beta \in (0, \infty)$. However, the extension to complex values is quite customary, as is treating $\{z(x) \mid x \in \X\}$ as a set of (commuting) abstract algebraic variables and each $z^X$ with $X \Subset \X$ as a multivariate monomial.

Furthermore, we fix an \emph{interaction}, by which we simply mean a function
\[
	W : \mathbf F \to \C,
\]
typically of the form $W = \exp(- \beta V)$ for some \emph{potential} $V : \mathbf F \to \R \cup \{+\infty\}$, where $\exp(-\infty) = 0$ by convention and $\beta$ is as before. Our use of $W$ follows Scott and Sokal (\cite{ref:scott-sokal}). We then define the \emph{Boltzmann factor}
\[
	\kappa : \mathbf F \to \C, \quad X \mapsto \prod_{S \subset X : S \neq \varnothing} W(S).
\]
Note the exclusion of the empty subconfiguration $\varnothing$ in the above product, normalising the Boltzmann factor at $\varnothing$ in the sense that $\kappa(\varnothing) = 1$. When $W$ has the aforementioned ``typical form'', $\kappa = \exp(- \beta H)$ where the \emph{Hamiltonian} $H : \mathbf F \to \R \cup \{+\infty\}$ is given by $X \mapsto \sum_{S \subset X : S \neq \varnothing} V(S)$.

\begin{remark}
	As is about to become apparent, the activity $z(x)$ of any given site $x \in \X$ and the corresponding singleton part $W(x) : = W(\{x\})$ of the interaction play equivalent roles in the sense that any of the two can be multiplicatively absorbed into the other without any meaningful impact on the model. Our definition may look clumsy but we think it justified by the following: (a) in the common view of partition functions as polynomials (or power series), $z$ corresponds to the indeterminates while the coefficients are based on $W = \e^{- \beta V}$; (b) for independence polynomials, $W$ is canonically $\{0, 1\}$-valued whereas $z$, of course, faces no such restriction; (c) $z$ may correspond to a natural a priori measure on $\X$, especially in continuum models, leaving $W$ to express local, possibly inhomogeneous modifications, be it model-specific or induced by boundary conditions.
\end{remark}

\subsection{Partition functions}

Let $\Lambda \Subset \X$ be a finite reference volume. The grand \emph{partition function} over $\Lambda$ is the quantity
\[
	Z(\Lambda) : = \sum_{X \subset \Lambda} z^X \kappa(X) \in \C.
\]
Note that the empty configuration always contributes the summand $\kappa(\varnothing) = Z(\varnothing) = 1$. Now let $x \in \X \setminus \Lambda$ be some additional site. Since $x$ can either be unoccupied or occupied, and only contributes to the configuration weight in the latter case, we have the \emph{fundamental identity}
\be \label{eq:fundamental-identity-Z+Z}
	Z(\{x\} \cup \Lambda) = Z(\Lambda) + Z(x, \Lambda),
\ee
where the second summand is the \emph{pinned partition function}
\[
	Z(x, \Lambda) : = \sum_{Y \subset \Lambda} z^{\{x\} \cup Y} \kappa(\{x\} \cup Y).
\]
If $Z(\Lambda) \neq 0$, the \emph{effective activity} (see remark below)
\[
	\widehat z (x, \Lambda) : = \frac{Z(x, \Lambda)}{Z(\Lambda)}
\]
of $x$ relative to $\Lambda$ is well-defined (note $\widehat z(x, \varnothing) = z(x) W(x)$) and we may rewrite \eqref{eq:fundamental-identity-Z+Z} as
\be \label{eq:fundamental-identity-Z(1+zhat)}
	Z(\{x\} \cup \Lambda) = Z(\Lambda)(1 + \widehat z(x, \Lambda)).
\ee
At this point, $Z(\{x\} \cup \Lambda) \neq 0$ is equivalent to $\widehat z(x, \Lambda) \neq -1$. This is the basis of the upcoming inductive argument.

\begin{remark}
	We adopt the terms ``fundamental identity'' and ``effective activity'' from \cite{ref:scott-sokal}. While only dealing with standard repulsive pair interactions, Section~3 of the latter reference contains these concepts as \cite[Eqs.~(3.3) and (3.45)]{ref:scott-sokal}, respectively. In fact, it is also instructive to look at \cite[Eqs.~(3.10) through (3.16)]{ref:scott-sokal}. Including our Eq.~\eqref{eq:fundamental-identity-Z(1+zhat)} as the second last, these identities encapsulate the specialisation of our later arguments to pair interactions. The length of Section~\ref{sec:proof} is mostly due to complications caused by the generalisation to multi-body interactions. See also \cite{ref:jansen-hierarchical, ref:jansen-neumann} for extensive use of effective activities in very specific lattice gases with hard-core pair interactions.
\end{remark}

\subsection{Dobrushin--Sokal criterion} \label{subsec:results-pf-non-zero}

A relatively simple approach to a criterion guaranteeing $Z(\Lambda) \neq 0$ for all finite reference volumes $\Lambda \Subset \X$ consists in bounding the effective activities, uniformly in their volume argument, by some fixed function $r : \X \to [0, 1)$. We also set $\alpha : = \frac{r}{1-r} : \X \to \R_+$ and trust the reader to be comfortable with shorthand like $\alpha^X = \prod_{x \in X} \alpha(x)$ for $X \Subset \X$, mirroring our notation of activity ``monomials''.

\begin{theorem} \label{thm:dobrushin-Z-non-zero-finite}
	Suppose that, for all $x \in \X$,
	\be \label{eq:dobrushin-criterion}
		|z(x)| \prod_{\substack{X \Subset \X :  \\ x \in X}} \max\{|W(X)|, 1 + |W(X) - 1| \alpha^S \mid \varnothing \neq S \subset X \setminus \{x\}\} \leq r(x).
	\ee
	Then $\sup_{\Lambda \Subset \X \setminus \{x\}} |\widehat z(x, \Lambda)| \leq r(x) < 1$ for all $x \in \X$. Hence, for all $\Lambda \Subset \X$,
	\[
		0 < (1-r)^\Lambda \leq |Z(\Lambda)| \leq (1+r)^\Lambda.
	\]
\end{theorem}

The factor for $X = \{x\}$ in \eqref{eq:dobrushin-criterion} is just $|W(x)|$, which can also be absorbed into $z(x)$ (see earlier remark). In all other factors, the initial term $|W(X)|$ in the maximum may seem irritating at first. Of course, it can be omitted in the \emph{repulsive} case where $|W(X)| \leq 1$. In particular, if $W$ is a standard real repulsive pair interaction in that $W(X) = W(X) \1_{\{|X| = 2\}} + \1_{\{|X| \neq 2\}} \in [0, 1]$ for all $X \Subset \X$, the inequality \eqref{eq:dobrushin-criterion} reduces to
\[
	|z(x)| \prod_{x' \in \X \setminus \{x\}} \frac{1 - W(\{x, x'\})r(x')}{1 - r(x')} \leq r(x),
\]
which is just condition (5.7) in \cite[Corollary~5.2]{ref:scott-sokal}, attributed there to Dobrushin and Sokal, respectively for hard-core (\cite{ref:dobrushin-semi-invariants, ref:dobrushin-saint-flour}) and general standard pair repulsions (\cite{ref:sokal}).

Conceptually, the insertion of $|W(X)|$ into the maximum mirrors generalisations of similar criteria from repulsive to \emph{stable} interactions. More precisely, \eqref{eq:dobrushin-criterion} implies the local stability condition
\[
	|z(x) W(x)| \prod_{X \Subset \X : \{x\} \subsetneq X} \max\{|W(X)|, 1\} \leq r(x) < 1,
\]
where the maxima inside the product are typically expressed via
\[
	\max\{\exp(-\beta V), 1\} = \exp(\max\{-\beta V, 0\}) \leq \exp(|\beta V|)
\]
for $W = \exp(-\beta V)$ with $\beta V \in \R \cup \{+\infty\}$. In particular, the hypothesis of Theorem~\ref{thm:dobrushin-Z-non-zero-finite} implies the bound
\[
	|z^{\{x\} \cup Y} \kappa(\{x\} \cup Y)| \leq r(x) |z^Y \kappa(Y)|
\]
for all $x \in \X$ and $Y \Subset \X \setminus \{x\}$, which also inductively yields the summand-wise second inequality in
\[
	|Z(\Lambda)| \leq \sum_{X \subset \Lambda} |z^X \kappa(X)| \leq \sum_{X \subset \Lambda} r^X = (1 + r)^\Lambda
\]
for all $\Lambda \Subset \X$.

\subsection{Further derived and comparable criteria}

Bringing the denominator of $r = \frac{\alpha}{1+\alpha}$ from the right-hand side of \eqref{eq:dobrushin-criterion} to the left and using the trivial estimate $1 + \zeta \leq \e^\zeta$ for $\zeta \geq 0$, we can derive a criterion requiring
\begin{align}
	& |z(x) W(x)| \prod_{X \Subset \X : \{x\} \subsetneq X} \max\{|W(X)|, 1\} \notag \\
	& \quad \times \exp \left( \alpha(x) + \sum_{X \Subset \X : \{x\} \subsetneq X} |W(X) - 1| \max_{S \subset X \setminus \{x\} : S \neq \varnothing} \{\alpha^S\} \right) \leq \alpha(x) \label{eq:kotecky-preiss-like-criterion}
\end{align}
for all $x \in \X$. If $W$ is a pair interaction, i.e., $W(X) = 1$ for all $X \Subset \X$ with $|X| \neq 2$, then \eqref{eq:kotecky-preiss-like-criterion} simplifies to
\begin{align*}
	& |z(x)| \prod_{x' \in \X \setminus \{x\}} \max\{|W(\{x, x'\})|, 1\} \\
	& \quad \times \exp \left( \alpha(x) + \sum_{x' \in \X \setminus \{x\}} |W(\{x, x'\}) - 1| \alpha(x') \right) \leq \alpha(x),
\end{align*}
which is a variant of the criterion of Koteck{\'y} and Preiss. Their original condition, see \cite{ref:kotecky-preiss}, deals only with lattice gases with pairwise hard-core interactions. The above variant is, up to the different approach to stability, an instance of the extended version in \cite{ref:poghosyan-ueltschi}, allowing for stable pairwise interactions on general position spaces without the need for hard-core self-repulsion, see also \cite[Corollary~2.2]{ref:jansen-kolesnikov}. To see the comparison, choose $\alpha = |z| e^a$ and note that the initial summand inside the above exponential corresponds to the hard-core self-repulsion that is implicit in our present framework. For pairwise repulsions, the stability considerations are obsolete and it suffices to consult \cite{ref:ueltschi}.

There is also a well-established way of circumventing the need to find a suitable $\alpha$ for the Koteck{\'y}--Preiss criterion that we can adapt to \eqref{eq:kotecky-preiss-like-criterion}. Assume $\alpha$ constant and, if $W(X) \neq 1$ for some $X \Subset \X$ with $|X| > 2$, then additionally assume $\alpha \leq 1$. The argument of the exponential in \eqref{eq:kotecky-preiss-like-criterion} then becomes $C_W(x) \alpha$ with
\[
	C_W(x) : = 1 + \sum_{X \Subset \X : \{x\} \subsetneq X} |W(X) - 1|,
\]
the assumed finiteness of which is not uncommon in the context of Gibbs point processes, cf. \cite[Def.~4.1.2]{ref:ruelle} or \cite[Theorem~B.1]{ref:jansen-cluster}. Since the term $\alpha \exp(- C_W(x) \alpha)$ assumes its maximal value at $\alpha = C_W(x)^{-1} \leq 1$, we can use requirement \eqref{eq:kotecky-preiss-like-criterion} to derive the following, where $1/\infty = \infty^{-1} = 0$.

\begin{cor} \label{cor:kotecky-preiss-without-alpha}
	Suppose that, for all $x \in \X$,
	\[
		|z(x) W(x)| \prod_{X \Subset \X : \{x\} \subsetneq X} \max\{|W(X)|, 1\} \leq \frac{1}{\overline{C}_W \e},
	\]
	where $\overline{C}_W : = \sup_{y \in \X} \{C_W(y)\}$. Then the hypotheses of Theorem~\ref{thm:dobrushin-Z-non-zero-finite} are satisfied with constant $\alpha = \overline{C}_W^{-1} \leq 1$ and accordingly chosen $r = \frac{\alpha}{1+\alpha} \leq 1/2$.
\end{cor}

To the best of our knowledge, for multi-body interactions, the corollary is new. For pair interactions, the condition is well-known apart from the uncommon way of expressing the local stability term, cf.\ \cite[Theorem~4.2.3]{ref:ruelle} for example.

For a comparison with an existing criterion for multi-body interactions, let $W$ be given by a pointwise \emph{absolutely summable potential} $V : \mathbf F \to \C$ in the sense that $W = \exp(-V)$ (for convenience and without loss, $\beta = 1$) and
\[
	D_V(x) : = \sum_{X \Subset \X : \{x\} \subsetneq X} |V(X)| < \infty
\]
for every $x \in \X$. Setting $\alpha = 1$ identically, hence $r = 1/2$, the trivial bound $1 + |\e^\zeta - 1| \leq \e^{|\zeta|}$ for arbitrary $\zeta \in \C$ makes \eqref{eq:dobrushin-criterion} hold whenever
\[
	2 |z(x) W(x)| \exp(D_V(x)) \leq 1.
\]
This condition improves the one by Gallavotti and Miracle-Sol{\'e} in \cite{ref:gallavotti-miracle-sole} where $\e^{D_V} - 1$ is used inside the exponential instead of the smaller $D_V$. This improvement extends to the slight refinement by the latter authors together with Robinson, see \cite{ref:gallavotti-miracle-sole-robinson} or \cite[Theorem~4.2.7]{ref:ruelle}. In fact, our alternative proof of Theorem~\ref{thm:dobrushin-Z-non-zero-finite} in Section~\ref{sec:alternative} shows that our improvement perfectly fits the context in \cite{ref:gallavotti-miracle-sole}, i.e., bounding the Kirkwood--Salsburg operator.

Another criterion for the non-vanishing of lattice gas partition functions with multi-body interactions can be found in \cite[Section~7]{ref:procacci-scoppola}. It is possible to improve the use of $D_V$ there, much like with the Gallavotti--Miracle-Sol{\'e}(--Robinson) criterion. However, we neither address the criteria in \cite{ref:procacci-scoppola} nor the more recent ones in \cite{ref:nguyen-fernandez} any further at this point. Both rely on the polymer expansion, which is intimately related to the cluster expansion and therefore not discussed in this paper.

\section{Preliminaries} \label{sec:preliminaries}

\noindent Let us start by updating our notational framework.

\subsection{Conditioning}

We first extend the interaction to a function
\[
	W : \mathbf F \times \mathbf F \to \C
\]
by the prescription
\be \label{eq:W-conditional-def}
	(X, B) \mapsto W(X \mid B) : = \begin{cases}
		\prod_{C \subset B} W(X \cup C) & \text{if } X \cap B = \varnothing, \\
		0 & \text{if } X = \{x\}, x \in B, \\
		1 & \text{otherwise}.
	\end{cases}
\ee
Note that the second part of this case distinction makes the corresponding extension
\[
	\kappa : \mathbf F \times \mathbf F \to \C, \quad (X, B) \mapsto \kappa(X \mid B) : = \prod_{S \subset X : S \neq \varnothing} W(S \mid B),
\]
of the Boltzmann factor vanish for non-disjoint arguments. Although this makes the third case in \eqref{eq:W-conditional-def} mostly irrelevant, our choice slightly simplifies some upcoming arguments. Clearly, our initial definitions of $W$ and $\kappa$ are recovered by fixing $B = \varnothing$, in which case we sometimes suppress the new second argument, whose conceptual role as a prior \emph{boundary condition} starts with the following property of $\kappa$: For all $X, Y, B \Subset \X$ with $X \cap Y = \varnothing$,
\be \label{eq:kappa-conditionally-multiplicative}
	\kappa(X \cup Y \mid B) = \kappa(X \mid Y \cup B) \kappa(Y \mid B).
\ee
This conditional multiplicativity of $\kappa$ is easily verifiable from the definitions above and reveals, e.g., by setting $B = \varnothing$ here, that our bivariate extension describes a relative Boltzmann factor comprised only of weights beyond the internal interactions of a fixed boundary condition.

\subsection{Partition functions revisited}

The introduction of boundary conditions obviously extends to partition functions. Given $X, B \Subset \X$ and $\Lambda \Subset \X$, we define
\[
	Z(X, \Lambda \mid B) : = \sum_{Y \subset \Lambda \setminus X} z^{X \cup Y} \kappa(X \cup Y \mid B).
\]
As with the case $B = \varnothing$, we sometimes suppress the first argument if $X = \varnothing$. Note that our extension is consistent with our initial definitions. The only discrepancy is the slight notational abuse of dropping curly set braces for singleton (sub-)configurations.

Two observations are in order. Joining the explicit excision of the pinned subconfiguration $X$ with the implicit excision of the boundary condition $B$ and employing \eqref{eq:kappa-conditionally-multiplicative}, we have
\be \label{eq:Z-pinned-boundary-conversion}
	Z(X, \Lambda \mid B) = Z(X, \Lambda \setminus (X \cup B) \mid B) = z^X \kappa(X \mid B) Z(\Lambda \mid X \cup B).
\ee
Furthermore, the fundamental identity \eqref{eq:fundamental-identity-Z+Z} generalises to
\be \label{eq:fundamental-identity-general}
	Z(\Lambda_1 \cup \Lambda_2 \mid B) = \sum_{X \subset \Lambda_1} Z(X, \Lambda_2 \mid B)
\ee
whenever $\Lambda_1, \Lambda_2 \Subset \X$ are two reference volumes, preferably disjoint.

\begin{remark}
	One often defines conditional partition functions in such a way that the points shared by reference volume and boundary condition are removed from the latter rather than the former. Hence, our definition makes the term ``boundary condition'' a little misleading. For our convenience, we change neither argument and avoid writing out cumbersome terms like the middle one in \eqref{eq:Z-pinned-boundary-conversion}.
\end{remark}

\subsection{Correlations and effective activities}

Deviating slightly from the standard definition (see the remark below), if $Z(\Lambda \mid B) \neq 0$, we call the ratio
\[
	R(X, \Lambda \mid B) : = \frac{Z(X, \Lambda \mid B)}{Z(\Lambda \mid B)}
\]
the \emph{correlation} of $X \Subset \X$ over $\Lambda \Subset \X$ relative to $B \Subset \X$. While excising $B$ from $\Lambda$ still changes nothing, the relation of $X$ to $\Lambda$ is now relevant due to the unpinned partition function in the denominator. If, for example, $X$ is a singleton with unique element $x \in \X \setminus \Lambda$, then, on the one hand, we define
\[
	\widehat z(x, \Lambda \mid B) : = R(x, \Lambda \mid B),
\]
$\widehat z(x, \Lambda) = \widehat z(x, \Lambda \mid \varnothing)$ clearly being consistent with our earlier definition, whereas, on the other hand, we can apply the first part of \eqref{eq:Z-pinned-boundary-conversion} in the numerator and \eqref{eq:fundamental-identity-general} (more specifically \eqref{eq:fundamental-identity-Z+Z}) in the denominator to obtain
\be \label{eq:R=zhat/(1+zhat)}
	R(x, \{x\} \cup \Lambda \mid B) = \frac{Z(x, \Lambda \mid B)}{Z(\Lambda \mid B) + Z(x, \Lambda \mid B)} = \frac{\widehat z(x, \Lambda \mid B)}{1 + \widehat z(x, \Lambda \mid B)},
\ee
provided that $Z(\Lambda \mid B) \neq 0$ and $\widehat z(x, \Lambda \mid B) \neq -1$.

\begin{remark}
	In the usual jargon, the correlation $R(X, \Lambda \mid B)$, if well-defined, would be set to zero whenever $X \not\subset \Lambda$. We deviate from this convention to streamline our definitions of partition functions, correlations and effective activities.
\end{remark}

\section{Proof of Theorem~\ref{thm:dobrushin-Z-non-zero-finite}} \label{sec:proof}

\noindent We now adapt the argument of \cite{ref:bencs-buys} to arbitrary complex interactions. To establish the central recursion formula for effective activities, we fix some initial arguments $x \in \X$ and $\Lambda \Subset \X \setminus \{x\}$. After demonstrating that our criterion is tailored to withstand the modifications that $W$ is subjected to over the course of this recursion, Theorem~\ref{thm:dobrushin-Z-non-zero-finite} follows by induction. As alluded to earlier, the case of pairwise interaction is a lot easier, taking only about one page in \cite[Eqs.~(3.10) through (3.16)]{ref:scott-sokal} to derive the recursion formula. The multi-body case is more intricate.

\subsection{Interpolating between interactions} \label{subsec:step-1-interpolation}

Fix some total order $\preceq$ on
\[
	\mathbf F(x) : = \{X \Subset \X \mid x \in X\},
\]
denote by $\prec$ the relation without identity and define, for each $X \in \mathbf F(x)$, the reduced set $X' : = X \setminus \{x\}$ as well as an interaction $W_X$ given by
\[
	W_X(Y) : = \begin{cases}
		W(Y \mid x) = W(Y) W(\{x\} \cup Y) & \text{if } x \notin Y, \{x\} \cup Y \prec X, \\
		1 & \text{if } x \in Y \in \mathbf F(x) \setminus \{X\}, \\
		W(Y) & \text{otherwise},
	\end{cases}
\]
for all $Y \Subset \X$. The extension to a conditional interaction $W_X : \mathbf F \times \mathbf F \to \C$ follows \eqref{eq:W-conditional-def} by decorating every $W$ therein with the index $X$. Note that $W_X$ only allows $x$ to interact along one subconfiguration, namely $X$, that every $S \in \mathbf F(x)$ with $S \prec X$ has its interaction transferred to $S'$ and that no trace of interactions along any $T \in \mathbf F(x)$ with $X \prec T$ remains. Functions and quantities derived from $W_X$ are denoted $\kappa_X$, $Z_X(\ldots)$, etc.

It should be intuitively clear that bonds $X \in \mathbf F(x)$ with trivial interaction $W(X) = 1$ are negligible until some modification, by conditioning or otherwise, disrupts the latter triviality. Noting the restricted domain $\Lambda$, let us fix some $\mathbf E \subset \mathbf F(x)$ containing all immediately relevant bonds in the sense that
\[
	\mathbf E \supset \mathbf E(x, \Lambda) : = \{X \in \mathbf F(x) \mid W(X) \neq 1, X' \subset \Lambda\}.
\]
Admittedly, the case of strict inclusion is not particularly relevant here since we do not actually perturb $W$. Our restriction to $\mathbf E$ is also related to a persistent convention where we set products equal to $0$ if at least one factor is well-defined and vanishes, even if some other factors are ill-defined.

\begin{lemma} \label{lem:step-1-interpolation}
	Suppose that $Z(\Lambda) \neq 0$ and, unless $z(x) W(x) = 0$, assume also $Z_X(\Lambda) \neq 0$ for all $X \in \mathbf E$. Then
	\be \label{eq:step-1-interpolation}
		\widehat z(x, \Lambda) = z(x) W(x) \prod_{X \in \mathbf E} \frac{Z_X(\Lambda \mid x)}{Z_X(\Lambda)}.
	\ee
	In particular, $z(x) = 1$ then implies $\widehat z(x, \Lambda) = \prod_{X \in \mathbf E} \widehat z_X(x, \Lambda)$.
\end{lemma}

\begin{proof}
	By \eqref{eq:Z-pinned-boundary-conversion}, we can write
	\[
		\widehat z(x, \Lambda) = z(x) W(x) \frac{Z(\Lambda \mid x)}{Z(\Lambda)}.
	\]
	If $z(x) W(x) = 0$, then \eqref{eq:step-1-interpolation} follows trivially by our convention. Otherwise, we need to show that
	\[
		\frac{Z(\Lambda \mid x)}{Z(\Lambda)} = \prod_{X \in \mathbf E} \frac{Z_X(\Lambda \mid x)}{Z_X(\Lambda)},
	\]
	where all denominators are assumed to be non-zero. Note first that, for every $X \in \mathbf F(x) \setminus \mathbf E(x, \Lambda)$ and all $Y \subset \Lambda$, we have $\kappa_X(Y) = \kappa_X(Y \mid x)$ so we can without loss assume $\mathbf E$ finite. Similarly, $\mathbf E(x, \Lambda) = \varnothing$ implies $\kappa(Y) = \kappa(Y \mid x)$ for all $Y \subset \Lambda$, making the case of empty $\mathbf E$ trivial. In the remaining cases of interest, let us enumerate the elements of $\mathbf E$ as $X_1, \ldots, X_{|\mathbf E|}$ such that the ordering of indices coincides with that given by $\preceq$. Observe then that, for all $Y \subset \Lambda$, we have $\kappa(Y) = \kappa_{X_1}(Y)$ as well as
	\[
		\kappa_{X_j}(Y \mid x) = \begin{cases}
			\kappa_{X_{j+1}}(Y) & \text{if } j \in \{1, \ldots, |\mathbf E|-1\}, \\
			\kappa(Y \mid x) & \text{if } j = |\mathbf E|,
		\end{cases}
	\]
	which implies that
	\[
		\frac{Z(\Lambda \mid x)}{Z(\Lambda)} =  \prod_{j=1}^{|\mathbf E|} \frac{Z_{X_j}(\Lambda \mid x)}{Z_{X_j}(\Lambda)} = \prod_{X \in \mathbf E} \frac{Z_X(\Lambda \mid x)}{Z_X(\Lambda)}
	\]
	is just a telescoping product expansion. With this, \eqref{eq:step-1-interpolation} is proven.
	
	If $z(x) = 1$, we can use the observation that
	\be \label{eq:kappa_X(x)}
		\kappa_X(x) = \begin{cases}
			W(x) & \text{if } X = \{x\}, \\
			1 & \text{otherwise},
		\end{cases}
	\ee
	for all $X \in \mathbf F(x)$ and write the right-hand side of \eqref{eq:step-1-interpolation} as
	\[
		\prod_{X \in \mathbf E} \frac{z(x) \kappa_X(x) Z_X(\Lambda \mid x)}{Z_X(\Lambda)} = \prod_{X \in \mathbf E} \frac{Z_X(x, \Lambda)}{Z_X(\Lambda)} = \prod_{X \in \mathbf E} \widehat z_X(x, \Lambda).
	\]
	This concludes the proof.
\end{proof}

\begin{remark}
	Identity \eqref{eq:step-1-interpolation} is our general version of line~5 on \cite[page~8]{ref:bencs-buys}.
\end{remark}

\subsection{Removing the root}

The next step consists of rewriting the factors in \eqref{eq:step-1-interpolation}. In particular, $x$ is removed from the boundary condition in each numerator by explicitly accounting for the one interaction term it contributes. Recall also that $X' = X \setminus \{x\}$ for all $X \in \mathbf F(x)$.

\begin{lemma} \label{lem:step-2-removal}
	Let $X \in \mathbf F(x)$ such that $Z_X(\Lambda) \neq 0$. If $X \neq \{x\}$, then
	\be \label{eq:step-2-removal-individual}
		\frac{Z_X(\Lambda \mid x)}{Z_X(\Lambda)} = 1 + (W(X) - 1) R_X(X', \Lambda) \1_{\{X' \subset \Lambda\}}.
	\ee
	Consequently, $\widehat z_X(x, \Lambda) = z(x) (1 + (W(X) - 1) R_X(X', \Lambda) \1_{\{X' \subset \Lambda\}})$, which holds trivially in the case $X = \{x\}$.
\end{lemma}

\begin{proof}
	For now, assume $X \neq \{x\}$. Since, for all $Y \Subset \X \setminus \{x\}$,
	\[
		\kappa_X(Y \mid x) = \begin{cases}
			\kappa_X(Y) & \text{if } X' \not\subset Y, \\
			W(X) \kappa_X(Y) & \text{if } X' \subset Y,
		\end{cases}
	\]
	or equivalently $\kappa_X(Y \mid x) = \kappa_X(Y) + (W(X) - 1) \kappa_X(Y) \1_{\{X' \subset Y\}}$, we have
	\[
		Z_X(\Lambda \mid x) = Z_X(\Lambda) + (W(X) - 1) Z_X(X', \Lambda) \1_{\{X' \subset \Lambda\}},
	\]
	which is closely related to \cite[Eq.~(2.2)]{ref:galvin-et-al} or the last item in \cite[Lemma~2.2]{ref:bencs-buys}. Dividing the last identity by $Z_X(\Lambda)$ yields \eqref{eq:step-2-removal-individual}.
	
	In view of \eqref{eq:kappa_X(x)}, multiplying by $z(x) W_X(x) = z(x)$ leads to the identity
	\[
		\widehat z_X(x, \Lambda) = z(x) (1 + (W(X) - 1) R_X(X', \Lambda) \1_{\{X' \subset \Lambda\}}),
	\]
	cf.\ the second part of \eqref{eq:Z-pinned-boundary-conversion}. Alternatively, we can modify the above argument and start with the observation that, for all $Y \Subset \X \setminus\{x\}$,
	\[
		\kappa_X(\{x\} \cup Y) = \kappa_X(x) \kappa_X(Y \mid x) = \kappa_X(Y) + (W(X) - 1) \kappa_X(Y) \1_{\{X' \subset Y\}},
	\]
	trivially extending to the case $X = \{x\}$ in the form
	\[
		\kappa_{\{x\}}(\{x\} \cup Y) = W(x) \kappa_{\{x\}}(Y) = \kappa_{\{x\}}(Y) + (W(x) - 1) \kappa_{\{x\}}(Y).
	\]
	In particular, the latter case ultimately translates into
	\[
		\widehat z_{\{x\}}(x, \Lambda) = z(x) W(x) = z(x)(1 + (W(x) - 1)),
	\]
	noting that $R_{\{x\}}(\varnothing, \Lambda) = 1$.
\end{proof}

\begin{remark}
	The identity \eqref{eq:step-2-removal-individual} corresponds to lines 12, 13 on \cite[page~8]{ref:bencs-buys}.
\end{remark}

Applying Lemma~\ref{lem:step-2-removal} to all factors in \eqref{eq:step-1-interpolation} nets us the following.

\begin{cor} \label{cor:step-2-removal}
	Suppose that the hypotheses of Lemma~\ref{lem:step-1-interpolation} hold. Then
	\be \label{eq:step-2-removal-product}
		\widehat z(x, \Lambda) = z(x) \prod_{X \in \mathbf F(x) : X' \subset \Lambda} \left( 1 + (W(X) - 1) R_X(X', \Lambda) \right).
	\ee
\end{cor}

\begin{proof}
	As indicated already, combining Lemmas~\ref{lem:step-1-interpolation} and \ref{lem:step-2-removal} yields
	\[
		\widehat z(x, \Lambda) = z(x) \prod_{X \in \mathbf E : X' \subset \Lambda} \left( 1 + (W(X) - 1) R_X(X', \Lambda) \right),
	\]
	noting the trivial/conventional treatment of the factor for $X = \{x\}$. The factors in \eqref{eq:step-2-removal-product} with $X \notin \mathbf E(x, \Lambda)$ are made irrelevant by our convention.
\end{proof}

\begin{remark}
	Lemma~\ref{lem:step-2-removal} provides a special instance of the Kirkwood--Salsburg equation. The latter is derived in the beginning of Section~\ref{sec:alternative} and its direct use in our alternative approach bypasses the need for the interpolating interactions $W_X$, $X \in \mathbf F(x)$, entirely.
\end{remark}

\subsection{Factorising correlations}

In the third step, we express the previously obtained correlations as products of single-site correlations which are, in turn, rewritten in terms of effective activities via \eqref{eq:R=zhat/(1+zhat)}. The first part of this programme is achieved via the following general observation.

\begin{lemma} \label{lem:R-factorisation}
	Let $Y \Subset \X$, $\preceq$ an arbitrary total order on $Y$ and suppose that $Z(\Lambda \mid \{y' \in Y \mid y' \prec y\}) \neq 0$ for all $y \in Y$. Then
	\be \label{eq:R-factorisation}
		R(Y, \Lambda) = \prod_{y \in Y} R(y, \Lambda \mid \{y' \in Y \mid y' \prec y\}).
	\ee
\end{lemma}

\begin{proof}
	Noting that $\{y' \in Y \mid y' \prec y\} = \varnothing$ for $y = \min_\preceq(Y)$, all terms in \eqref{eq:R-factorisation} are well-defined. Imitating arguments from previous proofs, let us enumerate the elements of $Y$ as $y_1, \ldots, y_{|Y|}$ in increasing order with respect to $\preceq$ and write
	\begin{align*}
		R(Y, \Lambda) & = z^Y \kappa(Y) \frac{Z(\Lambda \mid Y)}{Z(\Lambda)} \\
		& = \prod_{j=1}^{|Y|} z(y_j) \kappa(y_j \mid \{y_1, \ldots, y_{j-1}\}) \frac{Z(\Lambda \mid \{y_1, \ldots, y_{j-1}\} \cup \{y_j\})}{Z(\Lambda \mid \{y_1, \ldots, y_{j-1}\})} \\
		& = \prod_{j=1}^{|Y|} R(y_j, \Lambda \mid \{y_1, \ldots, y_{j-1}\}),
	\end{align*}
	thereby obtaining \eqref{eq:R-factorisation}. The first and third of the above equalities are applications of the definition of correlations together with the second part of \eqref{eq:Z-pinned-boundary-conversion}. The intermediate transformation uses the definition of $z^Y$, the conditional multiplicativity \eqref{eq:kappa-conditionally-multiplicative} of $\kappa$ and the obvious expansion of $Z(\Lambda \mid Y) / Z(\Lambda)$ into a telescoping product.
\end{proof}

We now arrive at the recursion formula for effective activities generalising \cite[Lemma~2.4 and Theorem~2.7]{ref:bencs-buys}. In order to confine the formula to one line while still utilising our notational framework, we equip each $X \in \mathbf F(x)$ with a total order $\preceq$ on $X'$ and define
\[
	X_{\prec x'}' : = \{x'' \in X' \mid x'' \prec x'\}
\]
for all $x' \in X'$. The slightly abusive ubiquity of the symbol $\preceq$ for all our total orders should not serve as a major source of confusion.

\begin{prop} \label{prop:step-3-factorisation}
	Suppose that $Z(\Lambda) \neq 0$ and, unless $z(x) W(x) = 0$, assume also that both
	\[
		Z_X(\Lambda \setminus \{x'\} \mid X_{\prec x'}') \neq 0 \quad \text{and} \quad \widehat z_X(x', \Lambda \setminus \{x'\} \mid X_{\prec x'}') \neq -1
		\]
		for all $X \in \mathbf E$ and $x' \in X'$. Then $\widehat z(x, \Lambda)$ equals
	\be \label{eq:step-3-factorisation-total}
		 z(x) \prod_{\substack{X \in \mathbf F(x) : \\ X' \subset \Lambda}} \left( 1 + (W(X) - 1) \prod_{x' \in X'} \frac{\widehat z_X(x', \Lambda \setminus \{x'\} \mid X_{\prec x'}')}{1 + \widehat z_X(x', \Lambda \setminus \{x'\} \mid X_{\prec x'}')} \right).
	\ee
\end{prop}

\begin{proof}
	With our product convention in mind, we skip over the trivial case where $z(x) W(x) = 0$. We can also disregard every factor of the outer product in \eqref{eq:step-3-factorisation-total} with $X \notin \mathbf E(x, \Lambda)$. Given $X \in \mathbf E(x, \Lambda) \subset \mathbf E$, our hypothesis, together with \eqref{eq:fundamental-identity-Z(1+zhat)} and \eqref{eq:R=zhat/(1+zhat)}, yields
	\[
		Z_X(\Lambda \mid X_{\prec x'}') = Z_X(\Lambda \setminus \{x'\} \mid X_{\prec x'}') (1 + \widehat z_X(x', \Lambda \setminus \{x'\} \mid X_{\prec x'}')) \neq 0,
	\]
	with $Z_X(\Lambda) \neq 0$ following from choosing $x' = \min_\preceq(X')$, as well as
	\[
		R_X(x', \Lambda \mid X_{\prec x'}') = \frac{\widehat z_X(x', \Lambda \setminus \{x'\} \mid X_{\prec x'}')}{1 + \widehat z_X(x', \Lambda \setminus \{x'\} \mid X_{\prec x'}')}
	\]
	for all $x' \in X'$. Hence, the assumptions of Lemma~\ref{lem:R-factorisation} are met with $W_X$ and $X'$ instead of $W$ and $Y$, respectively, so we obtain
	\[
		R_X(X', \Lambda) = \prod_{x' \in X'} R_X(x', \Lambda \mid X_{\prec x'}') = \prod_{x' \in X'} \frac{\widehat z_X(x', \Lambda \setminus \{x'\} \mid X_{\prec x'}')}{1 + \widehat z_X(x', \Lambda \setminus \{x'\} \mid X_{\prec x'}')}.
	\]
	Varying $X$, we see that the hypotheses shared by Lemma~\ref{lem:step-1-interpolation} and Corollary~\ref{cor:step-2-removal} are met and \eqref{eq:step-3-factorisation-total} now finally follows by inserting the above factorisation of correlations into \eqref{eq:step-2-removal-product}.
	
	Although the claim is already proven, we reformulate the result again to continue facilitating a more direct comparison with \cite{ref:bencs-buys}. Assume that $\mathbf E$ is finite and
	\[
		\mathbf E = \{X_1, \ldots, X_{|\mathbf E|}\}, \quad X_j' = \{x_{j, 1}, \ldots, x_{j, |X_j'|}\} \subset \Lambda, \quad j \in \{1, \ldots, |\mathbf E|\},
	\]
	with all indices consistent with the respective total orders. We can then write \eqref{eq:step-3-factorisation-total} in the form
	\[
		z(x) \prod_{j=1}^{|\mathbf E|} \left( 1 + (W(X_j) - 1) \prod_{k=1}^{|X_j'|} \frac{\widehat z_{j, k}(x_{j, k}, \Lambda \setminus \{x_{j, 1}, \ldots, x_{j, k}\})}{1 + \widehat z_{j, k}(x_{j, k}, \Lambda \setminus \{x_{j, 1}, \ldots, x_{j, k}\})} \right),
	\]
	where, for all $j \in \{1, \ldots, |\mathbf E|\}$ and $k \in \{1, \ldots, |X_j'|\}$, the effective activity $\widehat z_{j, k}(\ldots)$ is formed with respect to the interaction extending the prescription
	\[
		W_{j, k}(Y) : = W_{X_j}(Y \mid \{x_{j, 1}, \ldots, x_{j, k-1}\}), \quad Y \Subset \X,
	\]
	and the corresponding reference volumes also explicitly reflect the effective excision of this implicit boundary condition.
\end{proof}

\begin{remark}
	For readers interested in or familiar with the original argument in \cite{ref:bencs-buys}, let us note that the aforementioned statements generalised here as Proposition~\ref{prop:step-3-factorisation} are technically connected via the so-called Weitz hypertree construction, the primary device of Bencs and Buys' argument we choose to ignore.
\end{remark}

\subsection{Checking compatibility of our criterion}

Recall that our criterion requires that, for every choice of $x \in \X$,
\be \label{eq:dobrushin-criterion-proof}
	|z(x)| \prod_{X \in \mathbf F(x)} \max\{|W(X)|, 1 + |W(X) - 1| \alpha^S \mid \varnothing \neq S \subset X \setminus \{x\}\} \leq r(x),
\ee
where $r : \X \to [0, 1)$ and $\alpha = \frac{r}{1 - r} : \X \to \R_+$ are fixed functions. If we want to inductively infer that all effective activities occurring throughout the recursive application of Proposition~\ref{prop:step-3-factorisation} are bounded by the corresponding values of $r$, then \eqref{eq:step-3-factorisation-total} seemingly makes the maxima in \eqref{eq:dobrushin-criterion-proof} unnecessary. If $r \geq 1/2$, hence $\alpha \geq 1$, there is no difference but, otherwise, the maxima are meant to ensure the stability of our criterion under the conditioning happening during the recursion.

\begin{lemma} \label{lem:criterion-conditional-stability}
	Let $B \Subset \X \setminus \{x\}$ and $Y \in \mathbf F(x)$ such that $Y \cap B = \varnothing$. Then
	\begin{align*}
		& \max\{|W(Y \mid B)|, 1 + |W(Y \mid B) - 1| \alpha^S \mid \varnothing \neq S \subset Y \setminus \{x\}\} \\
		& \quad \leq \prod_{\substack{X \in \mathbf F(x) : \\ X \setminus B = Y}} \max\{|W(X)|, 1 + |W(X) - 1| \alpha^S \mid \varnothing \neq S \subset X \setminus (\{x\} \cup B)\}.
	\end{align*}
\end{lemma}

\begin{proof}
	By definition of the conditional interaction via \eqref{eq:W-conditional-def},
	\[
		W(Y \mid B) = \prod_{C \subset B} W(Y \cup C) = \prod_{X \in \mathbf F(x) : X \setminus B = Y} W(X)
	\]
	and Mayer's trick of inserting $-1+1$ into each factor of the middle product yields the expansion
	\begin{align*}
		W(Y \mid B) - 1 & = \sum_{\mathcal C \subset \{C \subset B\} : \mathcal C \neq \varnothing} \prod_{C \in \mathcal C} (W(Y \cup C) - 1) \\
		& = \sum_{C \subset B} (W(Y \cup C) - 1) \prod_{D \subset B : C \prec D} W(Y \cup D),
	\end{align*}
	where the last expression simply groups the non-empty collections $\mathcal C$ from the previous line according to their minimal elements with respect to some arbitrary but fixed total order $\preceq$ on $\{C \subset B\}$.
	
	Now let $S \subset Y \setminus \{x\}$. Then the above implies
	\[
		|W(Y \mid B) - 1| \alpha^S \leq \sum_{C \subset B} |W(Y \cup C) - 1| \alpha^S \prod_{D \subset B : C \prec D} |W(Y \cup D)|
	\]
	and, via a straightforward intermediate bound using maxima, we can reverse Mayer's trick to obtain
	\begin{align*}
		1 + |W(Y \mid B) - 1| \alpha^S & \leq \prod_{C \subset B} \max\{|W(Y \cup C)|, 1 + |W(Y \cup C) - 1| \alpha^S\} \\
	& = \prod_{X \in \mathbf F(x) : X \setminus B = Y} \max\{|W(X)|, 1 + |W(X) - 1| \alpha^S\}
	\end{align*}
	Applying this bound to non-empty $S$ complements the first observation of the proof in a way that clearly implies the claim.
\end{proof}

Taking the product over the possible choices of $Y$ in Lemma~\ref{lem:criterion-conditional-stability} and using \eqref{eq:W-conditional-def}, notably cases two and three thereof, we obtain the following.

\begin{cor} \label{cor:criterion-conditional-stability}
	Let $B \Subset \X$. Then
	\begin{align*}
		& \prod_{Y \in \mathbf F(x)} \max\{|W(Y \mid B)|, 1 + |W(Y \mid B) - 1| \alpha^S \mid \varnothing \neq S \subset Y \setminus \{x\}\} \\
		& \quad \leq \1_{\{x \notin B\}} \prod_{X \in \mathbf F(x)} \max\{|W(X)|, 1 + |W(X) - 1| \alpha^S \mid \varnothing \neq S \subset X \setminus \{x\}\}.
	\end{align*}
\end{cor}

Of course, we must also guarantee that our criterion survives the transition from $W$ to $W_X$ for arbitrary $X \in \mathbf F(x)$ (with $W(X) \neq 1$).

\begin{cor} \label{cor:criterion-interpolation-stability}
	Let $X \in \mathbf F(x)$. Then, for all $y \in \X$,
	\begin{align*}
		& \prod_{Y \in \mathbf F(y)} \max\{|W_X(Y)|, 1 + |W_X(Y) - 1| \alpha^S \mid \varnothing \neq S \subset Y \setminus \{y\}\} \\
		& \quad \leq \prod_{\substack{Y \in \mathbf F(y) : \\ x \in Y \Rightarrow Y \preceq X}} \max\{|W(Y)|, 1 + |W(Y) - 1| \alpha^S \mid \varnothing \neq S \subset Y \setminus \{y\}\}.
	\end{align*}
\end{cor}

\begin{proof}
	Fix $y \in \X$ and recall that, for all $Y \Subset \X$,
	\[
		W_X(Y) = \begin{cases}
			W(Y \mid x) & \text{if } Y \notin \mathbf F(x), \{x\} \cup Y \prec X, \\
			1 & \text{if } Y \in \mathbf F(x) \setminus \{X\}, \\
			W(Y) & \text{otherwise}.
		\end{cases}
	\]
	Using Lemma~\ref{lem:criterion-conditional-stability} to deal with the first case, we obtain the partial bound
	\begin{align*}
		& \prod_{Y \in \mathbf F(y) \setminus \mathbf F(x)} \max\{|W_X(Y)|, 1 + |W_X(Y) - 1| \alpha^S \mid \varnothing \neq S \subset Y \setminus \{y\}\} \\
		& \quad \leq \prod_{\substack{Y \in \mathbf F(y) : \\ x \in Y \Rightarrow Y \prec X}} \max\{|W(Y)|, 1 + |W(Y) - 1| \alpha^S \mid \varnothing \neq S \subset Y \setminus \{y\}\}.
	\end{align*}
	At most one non-trivial factor, one with $Y = X$, is missing here.
\end{proof}

\begin{remark}
	The above ``stability assertions'' correspond to the casual observation in \cite{ref:bencs-buys} that the vertex degrees, giving explicit choices of $r$ and $\alpha$, are not increased by the modifications there, cf.\ also our Corollary~\ref{cor:bencs-buys-recursion-stability}.
\end{remark}

\subsection{Concluding the proof}

We now assemble the results of the preceding subsections into the inductive proof of Theorem~\ref{thm:dobrushin-Z-non-zero-finite} so, for the remainder of this section, we assume \eqref{eq:dobrushin-criterion-proof} to hold for every possible choice of $x \in \X$.

Starting from the trivial observation $Z(\varnothing) = 1$, we inductively assume $\Lambda \Subset \X \setminus \{x\}$ to be such that $Z(\Lambda) \neq 0$ or, more precisely,
\[
	0 < (1 - r)^\Lambda \leq |Z(\Lambda)| \leq (1 + r)^\Lambda.
\]
In particular, $\widehat z(x, \Lambda)$ is well-defined and, by the fundamental identity \eqref{eq:fundamental-identity-Z(1+zhat)},
\[
	Z(\{x\} \cup \Lambda) = Z(\Lambda)(1 + \widehat z(x, \Lambda)),
\]
completing the induction step comes down to proving $|\widehat z(x, \Lambda)| \leq r(x) < 1$. To that end, we want to employ Proposition~\ref{prop:step-3-factorisation} so we first observe that, for all $X \in \mathbf F(x)$ and every $x' \in X' = X \setminus \{x\}$, Corollaries~\ref{cor:criterion-conditional-stability} and \ref{cor:criterion-interpolation-stability} allow us to inductively assume
\[
	Z_X(\Lambda \setminus \{x'\} \mid X_{\prec x'}') \neq 0 \quad \text{and} \quad |\widehat z_X(x', \Lambda \setminus \{x'\} \mid X_{\prec x'}')| \leq r(x') < 1
\]
and, in particular,
\[
	\left |\frac{\widehat z_X(x', \Lambda \setminus \{x'\} \mid X_{\prec x'}')}{1 + \widehat z_X(x', \Lambda \setminus \{x'\} \mid X_{\prec x'}')} \right| \leq \frac{r(x')}{1 - r(x')} = \alpha(x').
\]
We may therefore use Proposition~\ref{prop:step-3-factorisation} and equate $\widehat z(x, \Lambda)$ with
\[
	z(x) \prod_{\substack{X \in \mathbf F(x) : \\ X' \subset \Lambda}} \left( 1 + (W(X) - 1) \prod_{x' \in X'} \frac{\widehat z_X(x', \Lambda \setminus \{x'\} \mid X_{\prec x'}')}{1 + \widehat z_X(x', \Lambda \setminus \{x'\} \mid X_{\prec x'}')} \right),
\]
where the shorthand notation implicitly depends on our fixed choice of $x$ and the factor for $X = \{x\}$ in the outer product is just $W(x)$. We then obtain
\[
	|\widehat z(x, \Lambda)| \leq |z(x)| |W(x)| \prod_{X \in \mathbf F(x) \setminus \{\{x\}\}} \left( 1 + |W(X) - 1| \alpha^{X \setminus \{x\}} \right) \leq r(x) < 1,
\]
the second to last bound following from \eqref{eq:dobrushin-criterion-proof}. As mentioned before, we can now finally infer $Z(\{x\} \cup \Lambda) \neq 0$ or, more precisely,
\[
	0 < (1 - r)^{\{x\} \cup \Lambda} \leq |Z(\{x\} \cup \Lambda)| \leq (1 + r)^{\{x\} \cup \Lambda}
\]
and Theorem~\ref{thm:dobrushin-Z-non-zero-finite} follows inductively by extending $\Lambda$ one site at a time.

\section{Application to hypergraph independence polynomials} \label{sec:independence-polynomial}

\noindent We now demonstrate how Galvin et al.'s \cite[Theorem~1 / Theorem~10]{ref:galvin-et-al} on zero-freeness of hypergraph independence polynomials (Corollary~\ref{cor:galvin-et-al} here) directly follows from our main theorem. We then provide a full proof of Bencs and Buys' refined result, \cite[Theorem~1.1]{ref:bencs-buys} (Theorem~\ref{thm:bencs-buys} here), whose original proof is the inspiration and basis for the constructions of Section~\ref{sec:proof}.

In this section, we assume that $W$ is a pure hard-core interaction in the sense that it and, by extension, $\kappa$ only take the values $0$ and $1$. We can therefore identify $W$ with the \emph{hypergraph}
\[
	\mathfrak h : = \{e \Subset \X \mid W(e) \neq 1\} = \{e \Subset \X \mid W(e) = 0\}.
\]
For proper comparison with \cite{ref:scott-sokal, ref:galvin-et-al, ref:bencs-buys}, we should note that we deviate from the standard terminology in which ``hypergraph'' would refer to the pair $(\X, \mathfrak h)$ so $\mathfrak h$ would just be the set of ``edges''. Furthermore, there are different conventions regarding the admissibility or treatment of empty or singleton edges but we stick with our previous definitions and call
\[
	Z(\Lambda) = \sum_{X \subset \Lambda} z^X \1_{\{\forall e \in \mathfrak h \setminus \{\varnothing\} : e \not \subset X\}}
\]
the \emph{independence polynomial} of $\mathfrak h$ over $\Lambda \Subset \X$.

\begin{remark}
	The independence polynomials are clearly indifferent to the removal of edges that are not minimal in $\mathfrak h \setminus \{\varnothing\}$ with respect to inclusion. While we refrain from incorporating this observation into our argument, let us note that, in the ``mostly irrelevant'' third case of \eqref{eq:W-conditional-def}, our definition of conditional interactions is also meant to accommodate this lack of nuance.
\end{remark}

For $x \in \X$, the role of $\mathbf F(x)$ is now roughly played by the set
\[
	\mathfrak h(x) : = \{e \in \mathfrak h \mid x \in e\} \subset \mathbf F(x)
\]
of edges \emph{incident} to $x$. More precisely, $\mathfrak h(x)$ will play the role of the set $\mathbf E$ introduced in Section~\ref{sec:proof}. The \emph{degree} of $x \in \X$ in $\mathfrak h$ is then
\[
	\deg_\mathfrak h(x) : = |\mathfrak h(x)|
\]
and we need no further preparation to prove the main result of Galvin et.\ al.\ in \cite{ref:galvin-et-al}, see also \cite[Corollary~5.3]{ref:scott-sokal} for purely pairwise hard-core interactions.

\begin{cor}[{\cite[Theorem~10]{ref:galvin-et-al}}] \label{cor:galvin-et-al}
	Let $\Delta \in [1, \infty)$ and suppose that, for all $x \in \X$ with $\{x\} \notin \mathfrak h$,
	\[
		\deg_\mathfrak h(x) \leq \Delta \quad \text{and} \quad |z(x)| \leq  \frac{\Delta^{\deg_\mathfrak h(x)}}{(\Delta + 1)^{\deg_\mathfrak h(x) + 1}}.
	\]
	Then, for all $\Lambda \Subset \X$,
	\[
		|Z(\Lambda)| \geq \left( \frac{\Delta}{\Delta + 1} \right)^{|\Lambda|} > 0.
	\]
\end{cor}

\begin{proof}
	Let $x \in \X$ and observe that the requirement \eqref{eq:dobrushin-criterion} now reads
	\[
		|z(x)| \1_{\{\{x\} \notin \mathfrak h\}} \prod_{e \in \mathfrak h(x)} \left( 1 + \max\{\alpha^S \mid \varnothing \neq S \subset e \setminus \{x\}\} \right) \leq r(x).
	\]
	As always, $r : \X \to [0, 1)$ and $\alpha = \frac{r}{1-r} : \X \to \R_+$. Let us now assume that these functions are constant and $\alpha \leq 1$, i.e., $r \leq 1/2$. Then the above inequality further simplifies to
	\[
		|z(x)| \1_{\{\{x\} \notin \mathfrak h\}} \leq r(1-r)^{\deg_\mathfrak h(x)}.
	\]
	Noting that $1 - (\Delta + 1)^{-1} = \frac{\Delta}{\Delta + 1}$, the latter requirement is satisfied with $r = (\Delta + 1)^{-1}$ by the corollary's hypotheses so an application of Theorem~\ref{thm:dobrushin-Z-non-zero-finite} completes the proof.
\end{proof}

Note that the mapping $[0, 1) \times \R_+ \to [0, 1)$, $(r, \Delta) \mapsto r(1-r)^\Delta$, is non-increasing with respect to its second argument and that
\[
	\sup\{ r (1 - r)^\Delta \mid r \in [0, 1/2]\} = r (1 - r)^\Delta|_{r = (\Delta + 1)^{- 1}} = \frac{\Delta^\Delta}{(\Delta + 1)^{\Delta + 1}}.
\]
for every $\Delta \geq 1$. Therefore, the choice $r = (\Delta + 1)^{-1}$ is essentially optimal within the confines of the last proof's argumentation.

\subsection{Further refinement}

There is an even sharper criterion preventing the independence polynomials from vanishing, also based solely on (maximal) vertex degrees. In the case of graphs with only size-two edges, it can be found in \cite[Corollary~5.7]{ref:scott-sokal}. For general hypergraphs, it is the first main result of Bencs and Buys in \cite{ref:bencs-buys}, already acknowledged in \cite{ref:galvin-et-al} and presented here in the following form.

\begin{theorem}[{\cite[Theorem~1.1]{ref:bencs-buys}}] \label{thm:bencs-buys}
	Let $\Delta \in [2, \infty)$ and suppose that, for all $x \in \X$ with $\{x\} \notin  \mathfrak h$,
	\[
		\deg_\mathfrak h(x) \leq \Delta \quad \text{and} \quad |z(x)| \begin{cases}
			< 1 & \text{if } \deg_\mathfrak h(x) = 0, \\
			< \frac{1}{\Delta} & \text{if } \deg_\mathfrak h(x) = 1, \\
			\leq \frac{(\Delta - 1)^{\deg_\mathfrak h(x) - 1}}{\Delta^{\deg_\mathfrak h(x)}} & \text{if } \deg_\mathfrak h(x) \geq 2.
		\end{cases}
	\]
	Then $Z(\Lambda) \neq 0$ for all $\Lambda \Subset \X$.
\end{theorem}

Our proof of this result is essentially analogous to our proof of Theorem~\ref{thm:dobrushin-Z-non-zero-finite}. We fix some choice of $x \in \X$ and employ the same implicitly $x$-dependent notation that went into Proposition~\ref{prop:step-3-factorisation}. Using that, in our present framework, $W(e) = \1_{\{e \notin \mathfrak h\}} = 1 - \1_{\{e \in \mathfrak h\}}$ for all $e \Subset \X$, the recursion formula simply reads as follows: Given $\Lambda \Subset \X \setminus \{x\}$,
\be \label{eq:zhat-recursion-hard-core}
	\widehat z(x, \Lambda) = z(x) \1_{\{\{x\} \notin \mathfrak h\}} \prod_{\substack{e \in \mathfrak h(x) : \\ e' \subset \Lambda}} \left( 1 - \prod_{x' \in e'} \frac{\widehat z_e(x', \Lambda \setminus \{x'\} \mid e_{\prec x'}')}{1 + \widehat z_e(x', \Lambda \setminus \{x'\} \mid e_{\prec x'}')} \right),
\ee
provided that some reasonable conditions are satisfied. With indexing similar to that presented at the end of the proof of Proposition~\ref{prop:step-3-factorisation}, Bencs and Buys first derive this recursion formula for linear hypertrees to then reiterate the same argument along the construction of the so-called \emph{Weitz hypertree} of an arbitrary hypergraph, see \cite[Section~2]{ref:bencs-buys}. As previously remarked, we do not use this reduction.

As for the proof of Theorem~\ref{thm:dobrushin-Z-non-zero-finite}, we need to ensure that the criterion in Theorem~\ref{thm:bencs-buys} is compatible with the modifications and conditioning of the interaction going into the right-hand side of \eqref{eq:zhat-recursion-hard-core}. We start with the conditioning aspect and set
\[
	\mathfrak h / B : = \{e \Subset \X \mid W(e \mid B) \neq 1\} = \{e \Subset \X \mid W(e \mid B) = 0\}
\]
for all $B \Subset \X$, noting that the conditional interaction is indeed $\{0, 1\}$-valued again. In fact, it follows directly from our definition of conditional interactions in \eqref{eq:W-conditional-def} that
	\[
		\mathfrak h / B = \{e \setminus B \mid e \in \mathfrak h\} \cup \{\{y\} \mid y \in B\}
	\]
	for all $B \Subset \X$. The notation is taken from \cite[page~70 (``Contraction'')]{ref:galvin-et-al}, Bencs and Buys would write $\mathfrak h \ominus B$ instead (second item in \cite[Definition~2.1]{ref:bencs-buys}). The following is now evident.

\begin{lemma} \label{lem:bencs-buys-conditional-stability}
	Let $B \Subset \X$. Then
	\[
		(\mathfrak h / B)(x) = \begin{cases}
			\{e \setminus B \mid e \in \mathfrak h(x)\} & \text{if } x \notin B, \\
			\{\{x\}\} & \text{if } x \in B.
		\end{cases}
	\]
	In particular, $\{x\} \notin (\mathfrak h / B)(x)$ necessitates both $\deg_{\mathfrak h / B}(x) \leq \deg_\mathfrak h(x)$ and $\{x\} \notin \mathfrak h$.
\end{lemma}

To address the compatibility with our ``modifications'', we set
	\[
		\mathfrak h_e : = \{b \Subset \X \mid W_e(b) \neq 1\} = \{b \Subset \X \mid W_e(b) = 0\}
	\]
	or, more explicitly,
	\[
		\mathfrak h_e = \{b' \mid e \succ b \in \mathfrak h(x)\} \cup (\mathfrak h \setminus (\mathfrak h(x) \setminus \{e\}))
	\]
	for all $e \in \mathfrak h(x)$. Together with Lemma~\ref{lem:bencs-buys-conditional-stability}, we arrive at the following, which plays a role analogous to that of Corollaries~\ref{cor:criterion-conditional-stability} and \ref{cor:criterion-interpolation-stability}.

\begin{cor} \label{cor:bencs-buys-recursion-stability}
	Let $e \in \mathfrak h(x)$, $x' \in e' = e \setminus \{x\}$ and $y \in \X$ such that $\{y\} \notin \mathfrak h_e / e_{\prec x'}'$. Then
	\[
		\deg_{\mathfrak h_e / e_{\prec x'}'}(y) \leq \deg_\mathfrak h(y)
	\]
	and, if additionally $x \neq y$ or $\{x\} \notin \mathfrak h$, then also $\{y\} \notin \mathfrak h$.
\end{cor}

\begin{proof}
	Lemma~\ref{lem:bencs-buys-conditional-stability} yields
	\[
		\{y\} \notin \mathfrak h_e \quad \text{and} \quad \deg_{\mathfrak h_e / e_{\prec x'}'}(y) \leq \deg_{\mathfrak h_e}(y) \leq \deg_\mathfrak h(y),
	\]
	where the last inequality follows from the explicit form of $\mathfrak h_e$. Also,
	\[
		\{y\} \in \mathfrak h \setminus \mathfrak h_e = \mathfrak h(x) \setminus \{e\}
	\]
	is only possible if $x = y$ and $\{x\} \in \mathfrak h$.
\end{proof}

We now prove Theorem~\ref{thm:bencs-buys} by an induction that is additionally inspired by the proof of \cite[Corollary~5.7]{ref:scott-sokal}.

\begin{proof}[Proof of Theorem~\ref{thm:bencs-buys}]
	Taking a quick look back at the proof of Theorem~\ref{thm:dobrushin-Z-non-zero-finite}, we only need to show that, for some arbitrary choice of $\Lambda \Subset \X \setminus \{x\}$ with $Z(\Lambda) \neq 0$, the effective activity $\widehat z(x, \Lambda)$ has an absolute value smaller than $1$. In fact, we are going to prove
	\[
		|\widehat z(x, \Lambda)| < \begin{cases}
			1 & \text{if } e' = e \setminus \{x\} \subset \Lambda \text{ for all } e \in \mathfrak h(x), \\
			\Delta^{-1} & \text{otherwise}.
		\end{cases}
	\]
	Our main tool is the recursion formula \eqref{eq:zhat-recursion-hard-core},
	\[
		\widehat z(x, \Lambda) = z(x) \1_{\{\{x\} \notin \mathfrak h\}} \prod_{\substack{e \in \mathfrak h(x) : \\ e' \subset \Lambda}} \left( 1 - \prod_{x' \in e'} \frac{\widehat z_e(x', \Lambda \setminus \{x'\} \mid e_{\prec x'}')}{1 + \widehat z_e(x', \Lambda \setminus \{x'\} \mid e_{\prec x'}')} \right),
	\]
	whose validity, see Proposition~\ref{prop:step-3-factorisation}, is again trivial or follows inductively.
	
	If $z(x) = 0$ or $\{x\} \in \mathfrak h$, then our induction step immediately concludes with $\widehat z(x, \Lambda) = 0$. The cases where $\{e \in \mathfrak h(x) \mid e' \subset \Lambda\} = \varnothing$ and, in particular, $\{x\} \notin \mathfrak h$ are resolved by the hypotheses of Theorem~\ref{thm:bencs-buys}, namely
	\[
		|\widehat z(x, \Lambda)| = |z(x)| \begin{cases}
			< 1 & \text{if } \deg_\mathfrak h(x) = 0, \\
			< \Delta^{-1} & \text{if } \deg_\mathfrak h(x) = 1, \\
			\leq \frac{(\Delta - 1)^{\deg_\mathfrak h(x) - 1}}{\Delta^{\deg_\mathfrak h(x)}} \leq \frac{\Delta - 1}{\Delta^2} < \Delta^{-1} & \text{if } \deg_\mathfrak h(x) \geq 2.
		\end{cases}
	\]
	Hence, we now assume $z(x) \neq 0$ and $\{x\} \notin \mathfrak h$ as well as $\{e \in \mathfrak h(x) \mid e' \subset \Lambda\} \neq \varnothing$. Briefly fixing $e \in \mathfrak h(x)$ and $x' \in e' \neq \varnothing$, observe that, by Corollary~\ref{cor:bencs-buys-recursion-stability},
	\[
		\{y\} \notin \mathfrak h \quad \text{and} \quad \deg_{\mathfrak h_e / e_{\prec x'}'}(y) \leq \deg_\mathfrak h(y)
	\]
	for all $y \in \X$ such that $\{y\} \notin \mathfrak h_e / e_{\prec x'}'$, so we can now inductively assume
	\[
		Z_e(\Lambda \setminus \{x'\} \mid e_{\prec x'}') \neq 0 \quad \text{and} \quad |\widehat z_e(x', \Lambda \setminus \{x'\} \mid e_{\prec x'}')| < \Delta^{-1},
	\]
	where, for the latter bound, we used that the edge
	\[
		e \setminus e_{\prec x'}' = \{x\} \cup \{x'' \in e' \mid x' \preceq x''\} \in (\mathfrak h_e / e_{\prec x'}')(x')
	\]
	is not fully contained in $\Lambda = \{x'\} \cup (\Lambda \setminus \{x'\})$. Using the resulting bound
	\[
		\left| \frac{\widehat z_e(x', \Lambda \setminus \{x'\} \mid e_{\prec x'}')}{1 + \widehat z_e(x', \Lambda \setminus \{x'\} \mid e_{\prec x'}')} \right| < \frac{\Delta^{-1}}{1 - \Delta^{-1}} = (\Delta - 1)^{-1}
	\]
	in each factor of the following inner product, we now have
	\[
		\prod_{\substack{e \in \mathfrak h(x) : \\ e' \subset \Lambda}} \left| 1 - \prod_{x' \in e'} \frac{\widehat z_e(x', \Lambda \setminus \{x'\} \mid e_{\prec x'}')}{1 + \widehat z_e(x', \Lambda \setminus \{x'\} \mid e_{\prec x'}')} \right|  < \prod_{\substack{e \in \mathfrak h(x) : \\ e' \subset \Lambda}} \left( 1 + (\Delta - 1)^{-|e'|} \right)
	\]
	where we stress the strictness of the first inequality, due to the assumed existence of $e \in \mathfrak h(x)$ with $\varnothing \neq e' \subset \Lambda$. Recalling that we also currently assume $z(x) \1_{\{\{x\} \notin \mathfrak h\}} \neq 0$, we obtain from \eqref{eq:zhat-recursion-hard-core} that
	\[
		|\widehat z(x, \Lambda)| < |z(x)| \prod_{\substack{e \in \mathfrak h(x) : \\ e' \subset \Lambda}} \left( 1 + (\Delta - 1)^{-|e'|} \right) \leq |z(x)| \left( \frac{\Delta}{\Delta - 1} \right)^{|\{e \in \mathfrak h(x) \mid e' \subset \Lambda\}|}.
	\]
	Inserting the hypotheses of Theorem~\ref{thm:bencs-buys} in the form
	\[
		|z(x)| \leq \left( \frac{\Delta - 1}{\Delta} \right)^{|\mathfrak h(x)|} (\Delta - 1)^{-1} = \left( \frac{\Delta - 1}{\Delta} \right)^{|\mathfrak h(x)| - 1} \Delta^{-1}
	\]
	then leads to
	\[
		|\widehat z(x, \Lambda)| < \begin{cases}
			(\Delta - 1)^{-1} \leq 1 & \text{if } e' \subset \Lambda \text{ for all } e \in \mathfrak h(x), \\
			\Delta^{-1} & \text{otherwise},
		\end{cases}
	\]
	which completes the derivation of the bound claimed at the beginning of the proof. In all of the considered cases, we can finally infer
	\[
		Z(\{x\} \cup \Lambda) = Z(\Lambda) (1 + \widehat z(x, \Lambda)) \neq 0
	\]
	via the fundamental identity \eqref{eq:fundamental-identity-Z(1+zhat)} and conclude that Theorem~\ref{thm:bencs-buys} follows by induction on $\Lambda$.
\end{proof}

\begin{remark}
	A crucial ingredient to the proof of Theorem~\ref{thm:bencs-buys} is the observation that the effective activities $\widehat z_e(x', \Lambda \setminus \{x'\} \mid e_{\prec x'}')$ on the right-hand side of \eqref{eq:zhat-recursion-hard-core} are based on hypergraph interactions where at least one edge incident to the root $x'$, namely $e \setminus e_{\prec x'}'$, cannot contribute by virtue of containing the site $x \notin \Lambda$. One could, in principle, refine our main result, Theorem~\ref{thm:dobrushin-Z-non-zero-finite}, by an analogous observation in the general recursion formula \eqref{eq:step-3-factorisation-total} but, without additional control of the interaction values $\{W(X) \mid X \in \mathbf F(x)\} \setminus \{1\}$, this seems somewhat pointless.
\end{remark}

\section{Alternative approach} \label{sec:alternative}

\noindent The final section of this article is dedicated to the derivation of our criterion from the Kirkwood--Salsburg hierarchy for lattice gas correlations. Although this hierarchy is a classical tool, let us stress the significant non-standard aspects of our treatment. Firstly, our unconventional definition of correlations, see Section~\ref{sec:preliminaries}, persists here for its convenient connection to effective activities.

Secondly, focusing on a Picard iteration for finite lattice reference volumes, we forego the classical framing (e.g., in \cite{ref:gallavotti-miracle-sole, ref:gallavotti-miracle-sole-robinson, ref:ruelle, ref:pastur, ref:brascamp}) in explicitly functional-analytic terms. Instead, we adapt the ansatz in \cite{ref:jansen-cluster, ref:jansen-kolesnikov} by stripping it of its cluster expansion angle. In particular, we need neither introduce any norms on Banach spaces nor argue for strict contractivity of any operator. The relevant properties of solutions to finite-volume Kirkwood--Salsburg equations, such as uniqueness and their relation to partition functions, are treated in an ad hoc fashion.

Thirdly and most importantly, the proof that the hypothesis of our main result, Theorem~\ref{thm:dobrushin-Z-non-zero-finite}, implies the existence of appropriately bounded solutions of the Kirkwood--Salsburg hierarchy hinges on a partition scheme for coverings of finite sets, reminiscent of the partition schemes behind so-called tree-graph (in)equalities, see, e.g., \cite{ref:sokal, ref:scott-sokal, ref:poghosyan-ueltschi, ref:jansen-cluster}. It is precisely this partition scheme that allows us to refine Gallavotti and Miracle-Sol{\'e}'s bounds for the Kirkwood--Salsburg operator in \cite{ref:gallavotti-miracle-sole} (the actual comparison is in Section~\ref{sec:results} above). These bounds are also found in \cite[Subsection~4.2.6]{ref:ruelle}, including the slight modification of the hierarchy made by Gallavotti, Miracle-Sol{\'e} and Robinson in \cite{ref:gallavotti-miracle-sole-robinson}.

\subsection{The Kirkwood--Salsburg hierarchy}

The titular hierarchy encompasses a set of linear identities for correlations over a common reference volume. The relevant linear operators are encoded in the kernel function $\gamma : \X \times \mathbf F \times \mathbf F \to \C$,
\[
	(s, N, B) \mapsto \gamma(s, N \mid B) : = \sum_{M \subset N} (-1)^{|N \setminus M|} \kappa(s \mid B \cup M),
\]
where we can technically ignore the case $(\{s\} \cup N) \cap B \neq \varnothing$. We also ignore the possible extension in the first argument as in \cite[Eq.~(11) and Table~1]{ref:brascamp}. As usual, we may omit an empty last argument of $\gamma$ from our notation. We provide a more explicit characterisation of $\gamma$ in terms of $W$ further below but the Kirkwood--Salsburg hierarchy as such simply hinges on the fact that, by M{\"o}bius inversion,
\be \label{eq:kappa-sum-gamma}
	\kappa(s \mid B \cup M) = \sum_{N \subset M} \gamma(s, N \mid B)
\ee
for all $s \in \X$, $M, B \Subset \X$. This implies that, given any $\Lambda \Subset \X$, 
\begin{align*}
	Z(\{s\} \cup T, \Lambda) & = \sum_{M \subset \Lambda \setminus T} z(s) \kappa(s \mid T \cup M) z^{T \cup M} \kappa(T \cup M) \\
	& = \sum_{N \subset \Lambda \setminus T} z(s) \gamma(s, N \mid T) \sum_{L \subset \Lambda \setminus (T \cup N)} z^{T \cup N \cup L} \kappa(T \cup N \cup L) \\
	& = \sum_{N \subset \Lambda \setminus T} z(s) \gamma(s, N \mid T) Z(T \cup N, \Lambda) 
\end{align*}
holds for all $s \in \X$, $T \Subset \X \setminus \{s\}$, and, provided that additionally $Z(\Lambda) \neq 0$, we immediately get
\be \label{eq:R-KS-MM-hierarchy}
	R(\{s\} \cup T, \Lambda) = \sum_{N \subset \Lambda \setminus T} z(s) \gamma(s, N \mid T) R(T \cup N, \Lambda).
\ee
This is the \emph{Kirkwood-Salsburg (KS) equation} and we need not even talk about thermodynamic limits or phase transitions in order to justify its use within this article. See also the subsequent remark.

\begin{prop} \label{prop:KS-uniqueness}
	Let $\Lambda \Subset \X$ and suppose that $\rho : \mathbf F \to \C$ satisfies
\be \label{eq:rho-ks-hierarchy}
	\rho(\{s\} \cup T) = \sum_{N \subset \Lambda \setminus T} z(s) \gamma(s, N \mid T) \rho(T \cup N)
\ee
for all $s \in \X$ and $T \Subset \X \setminus \{s\}$. Then there exists $\mu(\varnothing) \in \C$ such that $\rho(X) = Z(X, \Lambda) \mu(\varnothing)$ for all $X \Subset \X$. In particular, if $\rho(\varnothing) \neq 0$, then  $Z(\Lambda) \neq 0$ as well and
\[
	R(X, \Lambda) = \frac{Z(X, \Lambda)}{Z(\Lambda)} = \frac{\rho(X)}{\rho(\varnothing)}
\]
for all $X \Subset \X$.
\end{prop}

\begin{proof}
	Consider the weight function
	\[
		\mu : \mathbf F \to \C, \quad X \mapsto \sum_{Y \subset \Lambda \setminus X} (-1)^{|Y|} \rho(X \cup Y),
	\]
	which, by the inclusion-exclusion principle, is also characterised by prescribing $\sum_{Y \subset \Lambda \setminus X} \mu(X \cup Y) = \rho(X)$ for every $X \Subset \X$. The identity
	\be \label{eq:mu-gnz}
		\mu(\{s\} \cup T) = z(s) \kappa(s \mid T) \mu(T)
	\ee
	then likewise follows for every choice of $s \in \X$ and $T \Subset \X \setminus \{s\}$ from
	\begin{align*}
		\sum_{Y \subset \Lambda \setminus (\{s\} \cup T)} \mu(\{s\} \cup T \cup Y) & = \sum_{N \subset \Lambda \setminus T} z(s) \gamma(s, N \mid T) \sum_{\substack{M \subset \Lambda \setminus T : \\ N \subset M}} \mu(T \cup M) \\
		& = \sum_{M \subset \Lambda \setminus T} z(s) \sum_{N \subset M} \gamma(s, N \mid T) \mu(T \cup M) \\
		& = \sum_{M \subset \Lambda \setminus T} z(s) \kappa(s \mid T \cup M) \mu(T \cup M) \\
		& = \sum_{Y \subset \Lambda \setminus (\{s\} \cup T)} z(s) \kappa(s \mid T \cup Y) \mu(T \cup Y),
	\end{align*}
	where the first identity is just \eqref{eq:rho-ks-hierarchy} and the second to last equality follows from \eqref{eq:kappa-sum-gamma} (cf.\ also \cite[Eq.~(17)]{ref:brascamp}). Using \eqref{eq:kappa-conditionally-multiplicative}, we can now inductively apply \eqref{eq:mu-gnz} to obtain
	\[
		\mu(X) = z^X \kappa(X) \mu(\varnothing) \quad \text{and} \quad \rho(X) = \sum_{Y \subset \Lambda \setminus X} \mu(X \cup Y) = Z(X, \Lambda) \mu(\varnothing)
	\]
	for every $X \Subset \X$, which completes the proof.
\end{proof}

\begin{remark}
	The above proof illustrates the equivalence between solutions of the full set of KS equations and multiples of the Boltzmann--Gibbs distribution in finite volumes. The effect of our unconventional framework for correlations on the KS hierarchy can also be interpreted as follows: a solution to the KS equations here solves the usual KS equations for all finite boundary conditions simultaneously and the compatibility across distinct boundary conditions forces a consistent weighting proportional to their respective Boltzmann factors.
\end{remark}

\subsection{Fixed point problem and Picard iteration}

Fixing an arbitrary but persistent choice of $\Lambda \Subset \X$, it is common to facilitate the analysis of the KS hierarchy by condensing it into a fixed point equation of the form
\[
	\rho = K_\Lambda \rho
\]
with $\rho : \mathbf F \to \C$. To that end, we first fix an arbitrary \emph{selector} by which we mean a map $s : \mathbf F \setminus \{\varnothing\} \to \X$ with $s_X : = s(X) \in X$ and $X_s' : = X \setminus \{s_X\}$ for all non-empty $X \Subset \X$. Our version of the KS operator $K_\Lambda$ is then given by
\[
	K_\Lambda \rho : X \mapsto \begin{cases}
		\rho(\varnothing) & \text{if } X = \varnothing, \\
		\sum_{N \subset \Lambda \setminus X_s'} z(s_X) \gamma(s_X, N \mid X_s') \rho(X_s' \cup N) & \text{if } X \neq \varnothing,
	\end{cases}
\]
for all $\rho : \mathbf F \to \C$.

In view of Proposition~\ref{prop:KS-uniqueness}, we are primarily interested in solving the KS equations with a function $\rho : \mathbf F \to \C$ satisfying $\rho(\varnothing) \neq 0$ or, more specifically and without loss, $\rho(\varnothing) = 1$. This restriction of the KS equations is equivalent to their usual formulation as an inhomogeneous linear equation. Our ansatz works with the following Picard iteration one might expect from an approach based on the Banach fixed point theorem: starting with $\rho_0 : = \1_{\{\varnothing\}} : \mathbf F \to \C$, we define $\rho_{n+1} : = K_\Lambda \rho_n$ for all $n \in \N_0$.

\begin{lemma} \label{lem:picard-existence-independence-finite}
	The sequence $(\rho_n)_{n \in \N_0}$ is independent of the choice of the selector $s$ with $\rho_n(\varnothing) = 1$ for every $n \in \N_0$.
\end{lemma}

\begin{proof}
	We basically follow \cite[Lemma~2]{ref:pastur}. With a nod to the proof of Proposition~\ref{prop:KS-uniqueness}, we inductively assume
	\[
		\rho_n(X) = \rho_n(X) \1_{\{|X| \leq n\}} = \sum_{Y \subset \Lambda \setminus X} m_n(|X \cup Y|) z^{X \cup Y} \kappa(X \cup Y)
	\]
	for a given $n \in \N_0$ and all $X \Subset \X$, where $m_n = m_n \1_{\{\bullet \leq n\}} : \N_0 \to \C$ does not depend on the choice of $s$. For $n = 0$, this is obviously the case with $m_0 = \1_{\{\bullet = 0\}}$. Now, for non-empty $X \Subset \X$,
	\begin{align*}
		\rho_{n+1}(X) & = \sum_{N \subset \Lambda \setminus X_s'} z(s_X) \gamma(s_X, N \mid X_s') \rho_n(X_s' \cup N) \\
		& = \sum_{N \subset \Lambda \setminus X_s'} z(s_X) \gamma(s_X, N \mid X_s') \\
		& \quad \quad \quad \times \sum_{\substack{M \subset \Lambda \setminus X_s' : \\ N \subset M}} m_n(|X_s' \cup M|) z^{X_s' \cup M} \kappa(X_s' \cup M) \\
		& = \sum_{M \subset \Lambda \setminus X_s'} z(s_X) \kappa(s_X \mid X_s' \cup M) m_n(|X_s' \cup M|) z^{X_s' \cup M} \kappa(X_s' \cup M) \\
		& = \sum_{Y \subset \Lambda \setminus X} m_n(|X \cup M| - 1) z^{X \cup Y} \kappa(X \cup Y),
	\end{align*}
	where the first three identities are respectively due to $\rho_{n+1} = K_\Lambda \rho_n$, the induction hypothesis and \eqref{eq:kappa-sum-gamma}. It is easy to see that the induction hypothesis is adequately propagated with $m_{n+1} = m_{n+1} \1_{\{\bullet \leq n+1\}} : \N_0 \to \C$ given by
	\[
		k \mapsto \begin{cases}
			m_n(k-1) & \text{if } k \neq 0, \\
			1 - \sum_{X \subset \Lambda : X \neq \varnothing} m_n(|X|-1) z^X \kappa(X) & \text{if } k = 0,
		\end{cases}
	\]
	the latter case ensuring consistency with the prescription
	\[
		\rho_{n+1}(\varnothing) = (K_\Lambda \rho_n)(\varnothing) = \rho_n(\varnothing) = 1.
	\]
	Hence, the lemma follows by induction.
\end{proof}

\subsection{Monotone domination ansatz}

It follows from Lemma~\ref{lem:picard-existence-independence-finite} that, if $\rho_n$ converges pointwise to some function $\rho : \mathbf F \to \C$ as $n \to \infty$, then $\rho$ is not only a fixed point of $K_\Lambda$ but a solution to the full set of KS equations with $\rho(\varnothing) = 1$ and so, by Proposition~\ref{prop:KS-uniqueness}, $Z(\Lambda) \neq 0$ and $\rho = R(\bullet, \Lambda)$. We address this question of convergence via the general ansatz in \cite{ref:jansen-kolesnikov}.

For a non-negative function $\tilde \rho : \mathbf F \to \R_+$, we define $\tilde K_\Lambda \tilde \rho : \mathbf F \to \R_+$ by setting
\[
	(\tilde K_\Lambda \tilde \rho)(X) : = \begin{cases}
		\tilde \rho(\varnothing) & \text{if } X = \varnothing, \\
		\sum_{N \subset \Lambda \setminus X_s'} |z(s_X)| |\gamma(s_X, N \mid X_s')| \tilde \rho(X_s' \cup N) & \text{if } X \neq \varnothing,
	\end{cases}
\]
for all $X \Subset \X$. The following is our adaptation of \cite[Theorem~2.1]{ref:jansen-kolesnikov}.

\begin{prop} \label{prop:jansen-kolesnikov-domination}
	Suppose that there exists a finite non-negative function $\xi : \mathbf F \to \R_+$ with $\tilde K_\Lambda \xi \leq \xi$ and $\xi(\varnothing) = 1$. Then the function
	\[
		\rho : \mathbf F \to \C, \quad X \mapsto \lim_{n \to \infty} \rho_n(X),
	\]
	is well-defined and one has $\rho = K_\Lambda \rho$ with $\rho(\varnothing) = 1$ and $|\rho| \leq \xi$.
\end{prop}

\begin{proof}
	If $\chi : \mathbf F \to \C$ satisfies $|\chi| \leq \xi$ pointwise, then we also obtain the bounds $|(K_\Lambda \chi)(\varnothing)| = |\chi(\varnothing)| \leq \xi(\varnothing)$ and
	\begin{align*}
		|(K_\Lambda \chi)(X)| & \leq \sum_{N \subset \Lambda \setminus X_s'} |z(s_X) \gamma(s_X, N \mid X_s') \chi(X_s' \cup N)| \\
		& \leq \sum_{N \subset \Lambda \setminus X_s'} |z(s_X)| |\gamma(s_X, N \mid X_s')| \xi(X_s' \cup N) \\
		& = (\tilde K_\Lambda \xi)(X) \\
		& \leq \xi(X)
	\end{align*}
	for all non-empty $X \Subset \X$. Since $|\rho_0| = \1_{\{\varnothing\}} \leq \xi$, it follows by induction that $|\rho_n| \leq \xi$ for all $n \in \N_0$.
	
	In fact, let us consider, for given $n \in \N_0$, the function
	\[
		\tilde \rho_n : = \rho_0 + \sum_{k=1}^n |\rho_k - \rho_{k-1}| : \mathbf F \to \R_+.
	\]
	Obviously, $\tilde \rho_n(\varnothing) = 1$ and $|\rho_n| \leq \tilde \rho_n$. If we inductively assume $\tilde \rho_n \leq \xi$, which is still trivially the case for $n = 0$, then
	\[
		\tilde \rho_{n+1} = |K_\Lambda \rho_0| + \sum_{k=1}^n |K_\Lambda(\rho_k - \rho_{k-1})| \leq \tilde K_\Lambda \tilde \rho_n \leq \tilde K_\Lambda \xi \leq \xi.
	\]
	In particular, the finiteness of $\xi$ implies the pointwise existence of
	\[
		\rho = \lim_{n \to \infty} \rho_n = \rho_0 + \sum_{k=1}^\infty (\rho_k - \rho_{k-1}) : \mathbf F \to \C
	\]
	with $\rho(\varnothing) = 1$ and $|\rho| \leq \xi$. Clearly, the finite sums defining $K_\Lambda$ yield
	\[
		K_\Lambda \rho = \lim_{n \to \infty} K_\Lambda \rho_n = \lim_{n \to \infty} \rho_{n+1} = \rho.
	\]
	This concludes the proof.
\end{proof}

As illustrated in \cite{ref:jansen-kolesnikov}, several sufficient conditions for the convergence of the Picard iteration can be derived by adequately choosing the \emph{ansatz function} $\xi$ in Proposition~\ref{prop:jansen-kolesnikov-domination}. For the proof of our Theorem~\ref{thm:dobrushin-Z-non-zero-finite}, we choose
\be \label{eq:xi-definition}
	\xi : \mathbf F \to \R_+, \quad X \mapsto (\alpha \1_\Lambda + r \1_{\X \setminus \Lambda})^X = \prod_{x \in X \cap \Lambda} \alpha(x) \prod_{x \in X \setminus \Lambda} r(x),
\ee
where $\alpha : \X \to \R_+$ and $r = \frac{\alpha}{1+\alpha} : \X \to [0, 1)$ are the same as before. Note that this $\xi$ depends on $\Lambda$ but not on the choice of the selector $s$, even though this is technically allowed.

\subsection{Bounding the kernel}

In order to derive our criterion from Proposition~\ref{prop:jansen-kolesnikov-domination}, we now need to bound $\tilde K_\Lambda \xi$. While our chosen definition of $\gamma$ immediately yields the KS equations, its well-known alternative characterisation is more useful when it comes to deriving bounds. Let $x \in \X$, $B \Subset \X$, and observe first that
\begin{align}
	\gamma(x, N \mid B) & = \sum_{M \subset N \setminus \{x\}} (-1)^{|N \setminus M|} \kappa(x \mid B \cup M) \notag \\
	& = \1_{\X \setminus B}(x) (-1)^{\1_N(x)} \gamma(x, N \setminus \{x\} \mid B) \label{eq:gamma-(-1)-gamma}
\end{align}
for all $N \Subset \X$ by the properties of our conditional Boltzmann factor. Regarding the relevant remaining cases, observe that, for all $M \Subset \X \setminus (\{x\} \cup B)$,
\begin{align*}
	\kappa(x \mid B \cup M) & = \prod_{L \subset M} (1 + W(\{x\} \cup L \mid B) - 1 ) \\
	& = \sum_{\mathcal L \subset \{L \subset M\}} \prod_{L \in \mathcal L} ( W(\{x\} \cup L \mid B) - 1 ).
\end{align*}
Grouping the indexing sets $\mathcal L$ by their unions $\underline{\mathcal L} : = \bigcup_{L \in \mathcal L} L$, the relation \eqref{eq:kappa-sum-gamma}, by M{\"o}bius inversion, necessitates that
\be \label{eq:gamma-cover-sum}
	\gamma(x, N \mid B) = \sum_{\mathcal C \subset \{L \subset N\} : \underline{\mathcal C} = N} \prod_{L \in \mathcal C} ( W(\{x\} \cup L \mid B) - 1 )
\ee
for all $N \Subset \X \setminus (\{x\} \cup B)$. Since the empty set cannot contribute to unions with other sets, we can also write \eqref{eq:gamma-cover-sum} as
\be \label{eq:gamma-kappa-gammahat}
	\gamma(x, N \mid B) = \kappa(s \mid B) \widehat \gamma(x, N \mid B)
\ee
with $\kappa(x \mid B) = W(x \mid B)$ as usual and
\[
	\widehat \gamma(x, N \mid B) = \sum_{\substack{\mathcal C \subset \{L \subset N \mid L \neq \varnothing\} : \\ \underline{\mathcal C} = N}} \prod_{L \in \mathcal C} ( W(\{x\} \cup L \mid B) - 1 ),
\]
cf.\ \cite[Eq.~(12)]{ref:gallavotti-miracle-sole} or \cite[Eq.~(4.2.62)]{ref:ruelle}.

\begin{remark}
	Let $X \in \mathbf F(x)$. Consider \eqref{eq:gamma-cover-sum} with $B = \varnothing$ and $W_X$ instead of $W$, see Subsection~\ref{subsec:step-1-interpolation}. Then the only $L \Subset \X \setminus \{x\}$ with
	\[
		W_X(\{x\} \cup L \mid B) = W_X(\{x\} \cup L) \neq 1
	\]
	is given by $L = X' = X \setminus \{x\}$. Hence, Lemma~\ref{lem:step-2-removal} is, as remarked there, just an application of the KS equation.
\end{remark}

As in the latter remark, recall that $\mathbf F(x) = \{X \Subset \X \mid x \in X\}$ for every $x \in \X$. The proof of the upcoming lemma is conceptually quite similar to that of Lemma~\ref{lem:criterion-conditional-stability}, the main difference being that simply sorting non-empty collections of subsets by their minimal subsets has to be replaced with a partition scheme for subset collections covering a given target set.

\begin{lemma} \label{lem:gamma-alpha-bound}
	Let $x \in \X$, $N \Subset \X \setminus \{x\}$. Then $|\gamma(x, N)| \alpha^N$ is at most
	\begin{align*}
		\sum_{\substack{\mathcal X \subset \mathbf F(x) : \\ \underline{\mathcal X} \setminus \{x\} = N}} \prod_{X \in \mathcal X} \left( \max\{|W(X)| , 1 + |W(X) - 1| \alpha^S \mid \varnothing \neq S \subset X \setminus \{x\}\} - 1 \right).
	\end{align*}
\end{lemma}

\begin{proof}
	Let $\preceq$ denote an arbitrary but fixed total order on the finite power set $\{L \subset N\}$ of $N$ and define, for every $\mathcal C \subset \{L \subset N\}$,
	\[
		\mathcal C^* : = \{L \in \mathcal C \mid L \not\subset \underline{\{L' \in \mathcal C \mid L' \prec L\}}\} = \{{\min}_\preceq\{L \in \mathcal C \mid n \in L\} \mid n \in \underline{\mathcal C}\}.
	\]
	Observe that, for every $\mathcal C \subset \{L \subset N\}$, we have $\underline{\mathcal C^*} = \underline{\mathcal C}$,
	\[
		\mathcal C^* \cap \{L \subset N \mid L \subset \underline{\{L' \in \mathcal C^* \mid L' \prec L\}}\} = \varnothing,
	\]
	which implies $\varnothing \notin \mathcal C^*$ in particular, and
	\begin{align*}
		& \{\mathcal C' \subset \{L \subset N\} \mid (\mathcal C')^* = \mathcal C^*\} \\
		& \quad = \{\mathcal C^* \cup \mathcal C'' \mid \mathcal C'' \subset \{L \subset N \mid L \subset \underline{\{L' \in \mathcal C^* \mid L' \prec L\}}\}\}.
	\end{align*}
	Sorting the indexing covers $\mathcal C$ of $N$ in the sum on the right-hand side of \eqref{eq:gamma-cover-sum} according to their ``minimal'' subcovers $\mathcal C^*$, we can derive
	\begin{align*}
		\gamma(x, N) \alpha^N & = \sum_{\substack{\mathcal C \subset \{L \cup N\} : \\ \underline{\mathcal C} = N, \mathcal C^* = \mathcal C}} \prod_{L \in \mathcal C} (W(\{x\} \cup L) - 1) \alpha^{L \setminus \underline{\{L' \in \mathcal C \mid L' \prec L\}}} \\
		& \quad \quad \times \prod_{L \subset N : L \subset \underline{\{L' \in \mathcal C^* \mid L' \prec L\}}} W(\{x\} \cup L).
	\end{align*}
	Noting that, in the above, the set $L \setminus \underline{\{L' \in \mathcal C \mid L' \prec L\}}$ is non-empty for all $L \in \mathcal C = \mathcal C^*$ and that $L = \varnothing$ only occurs in the second product, we obtain the bound
	\begin{align*}
		|\gamma(x, N)| \alpha^N & \leq \sum_{\substack{\mathcal C \subset \{L \cup N\} : \\ \underline{\mathcal C} = N, \mathcal C^* = \mathcal C}} \prod_{L \in \mathcal C} (m(L) - 1) \prod_{L \subset N : L \subset \underline{\{L' \in \mathcal C^* \mid L' \prec L\}}} m(L) \\
		& = \sum_{\substack{\mathcal C \subset \{L \cup N\} : \\ \underline{\mathcal C} = N}} \prod_{L \in \mathcal C} (m(L) - 1),
	\end{align*}
	where we reversed the partitioning of the indexing covers after setting
	\[
		m(L) = \max\{|W(\{x\} \cup L)|, 1 + |W(\{x\} \cup L) - 1| \alpha^S \mid \varnothing \neq S \subset L\}
	\]
	for every $L \subset N$. Mapping each $\mathcal C \subset \{L \cup N\}$ satisfying $\underline{\mathcal C} = N$ to the set
	\[
		\{\{x\} \cup L \mid L \in \mathcal C\} \in \{\mathcal X \subset \mathbf F(x) \mid \underline{\mathcal X} \setminus \{x\} = N\}
	\]
	completes the proof.
\end{proof}

Summing over different choices of $N$ in Lemma~\ref{lem:gamma-alpha-bound}, after possibly employing \eqref{eq:gamma-(-1)-gamma}, then easily yields the following.

\begin{cor} \label{cor:sum-gamma-alpha-bound}
	Let $x \in \X$. Then
	\[
		\sum_{N \subset \Lambda} |\gamma(x, N)| \alpha^N = (1 + \alpha(x) \1_\Lambda(x)) \sum_{N \subset \Lambda \setminus \{x\}} |\gamma(x, N)| \alpha^N
	\]
	with 
	\begin{align*}
		& \sum_{N \subset \Lambda \setminus \{x\}} |\gamma(x, N)| \alpha^N \\
		& \quad \leq \prod_{\substack{X \in \mathbf F(x) : \\ X \setminus \{x\} \subset \Lambda}} \max\{|W(X)| , 1 + |W(X) - 1| \alpha^S \mid \varnothing \neq S \subset X \setminus \{x\}\}.
	\end{align*}
\end{cor}

Using $\widehat \gamma$ instead of $\gamma$, see \eqref{eq:gamma-kappa-gammahat}, the choice $\alpha = 1$ and the simple bound $1 + |\e^{-\zeta} - 1| \leq \e^{|\zeta|}$, $\zeta \in \C$, make Corollary~\ref{cor:sum-gamma-alpha-bound} imply Gallavotti and Miracle-Sol{\'e}'s \cite[Prop.~1]{ref:gallavotti-miracle-sole} whose derivation is also found on \cite[page~81]{ref:ruelle}.

\subsection{Alternative proof of Theorem~\ref{thm:dobrushin-Z-non-zero-finite}}

Let us now assume the hypothesis of our main result, i.e., that
\be \label{eq:dobrushin-criterion-proof-ks-epsilon}
	|z(x)| \prod_{X \in \mathbf F(x)} \max\{|W(X)|, 1 + |W(X) - 1| \alpha^S \mid \varnothing \neq S \subset X \setminus \{x\}\} \leq r(x).
\ee
for all $x \in \X$. Given our fixed choice of $\Lambda \Subset \X$ and the accordingly defined ansatz function $\xi$, see \eqref{eq:xi-definition}, we claim that $(\tilde K_\Lambda \xi)(X) \leq \xi(X)$ for every non-empty $X \Subset \X$, regardless of the chosen selector $s$. Indeed, the multiplicative form of $\xi$ in terms of $\alpha$ and $r = \frac{\alpha}{1+\alpha}$ makes it sufficient to check
\[
	|z(x)| \sum_{N \subset \Lambda \setminus X'} |\gamma(x, N \mid X')| \alpha^N \leq r(x) (1 + \alpha(x) \1_\Lambda(x))
\]
for every $x \in \X$ and $X' \Subset \X \setminus \{x\}$, which follows by Corollary~\ref{cor:sum-gamma-alpha-bound} and \eqref{eq:dobrushin-criterion-proof-ks-epsilon}. Lemma~\ref{lem:criterion-conditional-stability}/Corollary~\ref{cor:criterion-conditional-stability} takes care of the boundary condition $X'$.

Proposition~\ref{prop:jansen-kolesnikov-domination} provides the Picard iteration with a well-defined pointwise limit $\rho = \lim_{n \to \infty} \rho_n : \mathbf F \to \C$. It satisfies $\rho(\varnothing) = 1$, $|\rho| \leq \xi$ and
\[
	K_\Lambda \rho = \lim_{n \to \infty} K_\Lambda \rho_n = \lim_{n \to \infty} \rho_{n+1} = \rho.
\]
As a consequence of Lemma~\ref{lem:picard-existence-independence-finite}, neither $\rho$ nor the validity of the latter identities depend on the choice of the selector $s$. In other words, $\rho$ also satisfies the hypothesis of Proposition~\ref{prop:KS-uniqueness}, thereby yielding $Z(\Lambda) \neq 0$ and, in particular,
\[
	|\widehat z(x, \Lambda)| = |R(x, \Lambda)| = |\rho(x)| \leq \xi(x) = r(x) < 1
\]
for every $x \in \X \setminus \Lambda$. By \eqref{eq:fundamental-identity-Z(1+zhat)} and the arbitrariness of $\Lambda \Subset \X$, this inductively implies the bounds
\[
	0 < (1 - r)^\Lambda \leq |Z(\Lambda)| \leq (1 + r)^\Lambda
\]
as in the first proof and we are done.

\subsubsection*{Statement on data availability and no conflict of interest} 
Data sharing is not applicable. We do not analyse or generate any datasets, because our work proceeds within a theoretical and mathematical approach. 

The author has no competing interests to declare that are relevant to the content of this article.

\subsubsection*{Acknowledgement.} This research has partially been funded by the Deutsche Forschungsgemeinschaft (DFG) by grant SPP 2265 ``Random Geometric Systems'', Project P13. Since the expiration of the latter grant, the author has been supported under Germany’s excellence strategy EXC-2111-390814868. The author would like to thank Sabine Jansen and Leonid Kolesnikov for helpful discussions.


\begin{thebibliography}{10}	
	\bibitem{ref:bencs-buys} Bencs,~F., Buys,~P.: Optimal zero-free regions for the independence polynomial of bounded degree hypergraphs. Random Struct.\ \& Algorithms \textbf{66}(4), e70018 (2025). \url{https://doi.org/10.1002/rsa.70018}
	
	\bibitem{ref:brascamp} Brascamp,~H.J.: The Kirkwood-Salsburg equations: Solutions and spectral properties. Commun.\ Math.\ Phys.\ \textbf{40}, 235–247 (1975). \url{https://doi.org/10.1007/BF01610000}
	
	\bibitem{ref:dobrushin-semi-invariants} Dobrushin,~R.L.: Estimates of semi-invariants for the Ising model at low temperatures. In: Dobrushin,~R.L., Minlos,~R.A., Shubin,~M.A., Vershik,~A.M., (eds.): Topics in Statistical and Theoretical Physics -- F.\ A.\ Berezin Memorial Volume. Am.\ Math.\ Soc.\ Transl., Ser.\ 2, \textbf{177}, 59–81 (1996). \url{https://doi.org/10.1090/trans2/177/05}
	
	\bibitem{ref:dobrushin-saint-flour} Dobrushin,~R.L.: Perturbation methods of the theory of Gibbsian fields. In: Bernard,~P., (ed.): Lectures on Probability Theory and Statistics. Lecture Notes in Mathematics, vol.\ \textbf{1648}, Ecole d'{\'E}t{\'e} de Probabilit{\'e}s de Saint-Flour XXIV - 1994, 1-66. Springer, Berlin, Heidelberg (1996). \url{https://doi.org/10.1007/BFb0095674}
	
	\bibitem{ref:nguyen-fernandez} Fernández,~R., Nguyen,~T.X.: High-temperature cluster expansion for classical and quantum spin lattice systems with multi-body interactions. J.\ Stat.\ Phys.\ \textbf{191}, 13 (2024). \url{https://doi.org/10.1007/s10955-024-03231-w}
	
	\bibitem{ref:gallavotti-miracle-sole} Gallavotti,~G., Miracle-Sol{\'e},~S.: Correlation functions of a lattice system. Commun.\ Math.\ Phys.\ \textbf{7}, 274–288 (1968). \url{https://doi.org/10.1007/BF01646661}
	
	\bibitem{ref:gallavotti-miracle-sole-robinson} Gallavotti,~G., Miracle-Sol{\'e},~S., Robinson,~D.W.: Analyticity properties of a lattice gas. Phys.\ Lett.\ A, \textbf{25}(7), 493-494 (1967). \url{https://doi.org/10.1016/0375-9601(67)90004-7}
	
	\bibitem{ref:galvin-et-al} Galvin,~D., McKinley,~G., Perkins,~W., Sarantis,~M., Tetali,~P.: On the zeroes of hypergraph independence polynomials. Comb., Probab.\ and Comput.\ \textbf{33}(1), 65-84 (2024). \url{https://doi.org/10.1017/S0963548323000330}
	
	\bibitem{ref:jansen-cluster} Jansen,~S.: Cluster expansions for Gibbs point processes. Adv.\ Appl.\ Probab.\ \textbf{51}(4),  1129–1178 (2019). \url{https://doi.org/10.1017/apr.2019.46} 
	
	\bibitem{ref:jansen-hierarchical} Jansen,~S.: Thermodynamics of a hierarchical mixture of cubes. J.\ Stat.\ Phys.\ \textbf{179}, 309–340 (2020). \url{https://doi.org/10.1007/s10955-020-02531-1}
	
	\bibitem{ref:jansen-kolesnikov} Jansen,~S., Kolesnikov,~L.: Cluster expansions: Necessary and sufficient convergence conditions. J.\ Stat.\ Phys.\ \textbf{189}, 33 (2022). \url{https://doi.org/10.1007/s10955-022-02992-6}
	
	\bibitem{ref:jansen-neumann} Jansen,~S., Neumann,~J.P.: Hierarchical cubes: Gibbs measures and decay of correlations. J.\ Stat.\ Phys.\ \textbf{191}, 161 (2024). \url{https://doi.org/10.1007/s10955-024-03375-9}
	
	\bibitem{ref:kotecky-preiss} Kotecký,~R., Preiss,~D.: Cluster expansion for abstract polymer models. Commun.\ Math.\ Phys.\ \textbf{103}(3), 491–498 (1986). \url{https://doi.org/10.1007/BF01211762}
	
	\bibitem{ref:lee-yang-i} Lee,~T.D., Yang,~C.N.: Statistical theory of equations of state and phase transitions. I. Theory of condensation. Phys.\ Rev.\ \textbf{87}(3), 404-409 (1952). \url{https://doi.org/10.1103/PhysRev.87.404}
	
	\bibitem{ref:lee-yang-ii} Lee,~T.D., Yang,~C.N.: Statistical theory of equations of state and phase transitions. II. Lattice gas and Ising model. Phys.\ Rev.\ \textbf{87}(3), 410-419 (1952). \url{https://doi.org/10.1103/PhysRev.87.410}
	
	\bibitem{ref:pastur} Pastur,~L.A.: Spectral theory of Kirkwood-Salzburg equations in a finite volume. Theor.\ Math.\ Phys.\ \textbf{18}, 165–171 (1974). \url{https://doi.org/10.1007/BF01035916}
	
	\bibitem{ref:poghosyan-ueltschi} Poghosyan,~S., Ueltschi,~D.: Abstract cluster expansion with applications to statistical mechanical systems. J.\ Math.\ Phys.\ \textbf{50}(5), 053509 (2009). \url{https://doi.org/10.1063/1.3124770}
	
	\bibitem{ref:procacci-scoppola} Procacci,~A., Scoppola,~B.: Polymer gas approach to $N$-body lattice systems. J.\ Stat.\ Phys.\ \textbf{96}, 49–68 (1999). \url{https://doi.org/10.1023/A:1004564214528}
		
	\bibitem{ref:ruelle} Ruelle,~D.: Statistical Mechanics: Rigorous Results. World Scientific, (1999). \url{https://doi.org/10.1142/4090}
	
	\bibitem{ref:scott-sokal} Scott,~A.D., Sokal,~A.D.: The repulsive lattice gas, the independent-set polynomial, and the Lovász local lemma. J.\ Stat.\ Phys.\ \textbf{118}, 1151–1261 (2005). \url{https://doi.org/10.1007/s10955-004-2055-4}
	
	\bibitem{ref:sokal} Sokal,~A.D.: Bounds on the complex zeros of (di)chromatic polynomials and Potts-model partition functions. Comb.\ Probab.\ Comput.\ \textbf{10}(1), 41-77 (2001). \url{https://doi.org/10.1017/S0963548300004612}
	
	\bibitem{ref:ueltschi} Ueltschi,~D.: Cluster expansions and correlation functions. Mosc.\ Math.\ J.\ \textbf{4}, 511-522 (2004). \url{https://doi.org/10.17323/1609-4514-2004-4-2-511-522}
\end{thebibliography}
\end{document}